\documentclass[10pt]{article} 
\usepackage{sectsty}

\chapternumberfont{\bfseries\Huge\sc} 
\chaptertitlefont{\huge\sc}

\usepackage[a4paper, margin=1in]{geometry} 

\usepackage[utf8]{inputenc}
\usepackage[T1]{fontenc}

\usepackage{algorithmic}
\usepackage{amsfonts}
\usepackage{amsmath}
\usepackage{amssymb} 
\usepackage{amsthm} 
\usepackage{anyfontsize}
\usepackage{booktabs} 
\usepackage{caption} 
\usepackage{color}
\usepackage{comment}
\usepackage{courier} 
\usepackage{dsfont}
\usepackage[shortlabels]{enumitem}
\usepackage{float}
\usepackage{graphicx} 
\usepackage{harpoon} 
\usepackage{latexsym}
\usepackage{mathtools}
\usepackage{qtree}
\usepackage{slashed} 
\usepackage{subcaption}
\usepackage{tikz-cd}
\usepackage{url}
\usepackage{xcolor}
\usepackage{xurl}
\usepackage{breqn}

\usepackage[utf8]{inputenc} 
\usepackage{hyperref}
\definecolor{darkblue}{rgb}{0, 0, 0.5}
\hypersetup{colorlinks=true,citecolor=darkblue, linkcolor=darkblue, urlcolor=darkblue}

\DeclareMathOperator{\Fock}{{\Gamma}}

\DeclareMathOperator{\FreeFockOld}{{\Fock}} 

\DeclareMathOperator{\Dom}{{Dom}}

\newcommand{\tsr}{\otimes}

\newcommand{\closure}[1]{\overline{#1}}

\newcommand{\inner}[1]{{\langle #1\rangle}}

\newcommand{\NN}{{\mathds{N}}}

\newcommand{\laplace}{\bigtriangleup}

\newcommand{\Diff}[1]{\dfrac{d}{d #1}}

\newcommand{\footurl}[1]{{\footnote{\href[#1]{#1}}}}

\newcommand{\one}{{I}} 

\newcommand{\norm}[1]{\left\lVert#1\right\rVert}
\newcommand{\vertiii}[1]{{\left\vert\kern-0.25ex\left\vert\kern-0.25ex\left\vert #1 \right\vert\kern-0.25ex\right\vert\kern-0.25ex\right\vert}}

\newtheorem{thm}{Theorem}[section]
\newtheorem{prop}[thm]{Proposition}
\newtheorem{statement}[thm]{Statement}

\newtheorem{theorem}[thm]{Theorem}

\newtheorem{lemma}[thm]{Lemma}

\newtheorem{corollary}[thm]{Corollary}

\theoremstyle{definition}

\newtheorem{definition}[thm]{Definition}

\newtheorem{example}[thm]{Example}
\newtheorem{property}[thm]{Property}
\newtheorem{observation}[thm]{Observation}
\newtheorem{remark}[thm]{Remark}

\newcommand{\pmat}[1]{\begin{pmatrix} #1\end{pmatrix}}

\newcommand{\eps}{\epsilon}

\newcommand{\Pc}{{\mathcal{P}}}

\newcommand{\Hc}{{\mathcal{H}}}

\newcommand{\Bc}{{\mathcal{B}}}

\newcommand{\Ic}{{\mathcal{I}}}
\newcommand{\Jc}{{\mathcal{J}}}
\newcommand{\Tc}{{\mathcal{T}}}

\newcommand{\Sc}{{\mathcal{S}}}
\newcommand{\Nc}{{\mathcal{N}}}

\newcommand{\Lc}{{\mathcal{L}}}

\newcommand{\C}{{\mathcal{C}}}

\newcommand{\CC}{\mathbb{C}}

\newcommand{\KK}{\mathbb{K}}

\newcommand{\RR}{\mathds{R}}

\newcommand{\ZZ}{\mathds{Z}}

\newcommand{\ZZnn}{\mathds{Z}_{\geq 0}}

\newcommand{\txt}[1]{\text{#1}}

\newcommand{\aln}[1]{\begin{align*}#1\end{align*}}
\newcommand{\alnn}[1]{\begin{align}#1\end{align}}

\newcommand{\FExp}{\textsc{E}}

\newcommand{\I}{{\mathrm{i}}} 

\newcommand{\smfrac}[2]{{\textstyle\frac{#1}{#2}}}

\newcommand{\addsubtitlequote}[2]{
{}
}

\DeclareMathOperator{\SymmSum}{{\Sc}}

\DeclareUnicodeCharacter{221E}{\ensuremath{\infty}}

\newcommand{\doubledot}{\cdot, \cdot}

\newcommand{\cblue}[1]{\textcolor{blue}{#1}}
\newcommand{\cred}[1]{\textcolor{red}{#1}}

\newcommand{\cviolet}[1]{\textcolor{violet}{#1}}

\newcommand{\SelfRef}[1]{ 
}

\newcommand{\UrlRef}[1]{\href{#1}{url}}

\newcommand{\Arxiv}[1]{{\href{#1}{arxiv}}} 
\newcommand{\Wiki}[1]{{\href{#1}{wiki}}} 
\newcommand{\StackEx}[1]{{\href{#1}{stackex}}} 

\newcommand{\Rem}[1]{\textcolor{violet}{#1}} 
\newcommand{\FRem}[1]{\footnote{\Rem{#1}}}
\newcommand{\FRemHistorical}[1]{}

\newcommand{\FOnlineRef}[1]{\FRem{\href{#1}{online}}} 

\newcommand{\DeleteThisMaybe}[1]{}
\newcommand{\SuppressForNow}[1]{}

\newcommand{\V}[1]{{\bf{#1}}}

\usepackage[style=numeric, 
natbib=true, 
backend=biber, 
giveninits=true, 
uniquename=init,
isbn=false,
doi=false,
eprint=false,
maxbibnames=3, 
maxcitenames=1 %
]{biblatex}

\newcommand{\CarlemanOp}[1]{\textsc{C}(#1)}

\newtheoremstyle{nonumber-nopunct}
  {3pt}          
  {3pt}          
  {} 
  {}             
  {\bfseries}    
  {}             
  {.5em}         
  {}             
\theoremstyle{nonumber-nopunct}
\newtheorem*{block*}{} 

\usepackage{amsthm}

\usepackage{endnotes}
\usepackage{bigfoot}

\newcommand{\CarlFk}[1]{\CarlemanOp{F_1^{#1}, F_2^{#1}}} 
\newcommand{\CarlW}{\CarlemanOp{W_1, W_2}}

\begin{document}

\title{Nonlinear semigroups with unbounded generators under Carleman linearization}
\author{S. Gakkhar\footnote{{Sitanshu Gakkhar, \href{mailto:sgakkhar@waterloo.ca}{email}}, \today}, Ala Shayeghi, 
David C. Del Rey Fern\'andez  
}
\maketitle

\begin{abstract}
We treat the convergence of Carleman linearization of nonlinear evolutionary equations through the approximation theory of strongly continuous semigroups, by Carleman embedding the underlying nonlinear semigroups as linear semigroups. Linear semigroup theory then lets one replace the norm constraint on the convergence of Carleman linearization in the form used by quantum algorithms for a class of semi-discretized evolution equations by a dissipativity constraint, simplifying arguments for convergence. Applying Trotter-Kato approximation theorem to the linearized semigroup realizes the semigroup as a limit finite dimensional operator exponentials, reducing the question of convergence rate  of Carleman linearization to that of the Trotter-Kato approximation.  We then examine convergence of the Carleman linearization as the operators become unbounded, treating the hyperviscuous Burger's equation as an example.
Next we consider the perturbation theory of the Carleman semigroup and obtain conditions when polynomial nonlinearities correspond to the Carleman linearized semigroup being a $1$-integrated semigroup, so convergence is implied by variants of Trotter-Kato approximation for integrated semigroups.
\end{abstract}

\section{Introduction}\label{section_introduction}

Carleman linearization has found frequent use in designing efficient quantum algorithms for nonlinear differential equations of type \alnn{\label{eq_orig_nonlinear_eqn}
\phi' = W_1\phi + W_2\phi^{\tsr 2},\ W_1:H\to H, W_2:H^{\tsr 2}\to H
} on a Hilbert space $H$ subject to initial condition $\phi(0)=\phi_0\in H$, as well as the inhomogeneous version $\phi' = W_0(t) + W_1\phi + W_2\phi^{\tsr 2}$.
It proceeds by deriving an equivalent infinite dimensional linear ordinary differential equation system on the diagonal inside $\Fock(H) := \oplus_n H^{\tsr n}$, \alnn{\label{eq_full_carleman_system}
\Diff{t}\pmat{\phi\\ \phi^{\tsr 2}\\ \vdots \\ \phi^{\tsr n}\\ \vdots} &= \pmat{
\SymmSum_1(W_1) &\SymmSum_1(W_2) &0 &0 &0 &\dots\\
0 &\SymmSum_2(W_1) &\SymmSum_2(W_2) &0 &0 &\dots \\
\vdots &\vdots\\
0 &\dots &\SymmSum_n(W_1) &\SymmSum_n(W_2) &0 &\dots \\
\vdots &\vdots\\
} \pmat{\phi\\ \phi^{\tsr 2}\\ \vdots \\ \phi^{\tsr n}\\ \vdots}\ ,\\
(\FExp(&\phi))' = \CarlemanOp{W_1,W_2}\FExp(\phi), \txt{ with } \FExp(\phi)(0)  = \FExp(\phi_0) \nonumber,
}
where $\FExp: H\to \Fock(H), \FExp(\phi) := \oplus_n \phi^{\tsr n}$, and $\SymmSum_k(W_i)$ are $W_i$'s symmetrized to act on $H^{\tsr k}, H^{\tsr k+1}$. Notice that $W_2$ is given as a linear operator defined on all of $H^{\tsr 2}$.
Therefore, $\CarlemanOp{W_1, W_2}$ defines a linear dynamical system on $\Fock(H)$.
For the linear system on $\Fock(H)$, which is a lift of the original system to $\Fock(H)$,
one wants that the truncations obtained by restricting to $\oplus_{n\leq N} H^{\tsr n}$ converge to the solution as $N\to \infty$ with respect to projection on the first tensor power $\Tc_1:\FreeFockOld(H)\to H$ in the sense that the solution
$\hat u_N$ on $\Fock(H)_N$
to the finite truncation and the solution $u(t)$ to the original system satisfy\aln{
\lim_{N\to \infty} \norm{\Tc_1 \hat u_N (t) - u(t)}_H = 0.
}

In the literature, establishing this convergence in $N$, which we will refer to as the \emph{level} of the Carleman linearization, requires that the $W_1$ must be diagonalizable with eigenvalues $\mu_k, k\in \NN$, strictly negative
and $R := \norm{W_2}\norm{\phi_0}/\max_k |\mu_k|<1$ must hold \cite{liu2021_dissipative_nonlinear_carleman}. In this case the error decays as $R^N$. However, one observes the following.

\begin{observation}\label{obv_r_bound} Applied to partial differential equations (PDEs), like $\phi' = L(\phi)$ for $L$, a partial differential operator, quadratically non-linear, one works by fixing a discretization of $L$ of form $W_1\phi + W_2\phi^{\tsr 2}$, but the convergence condition $R=\norm{W_2}\norm{\phi_0}/\max_k |\mu_k|<1$ cannot possibly hold in the limit of the discretization as $\norm{W_2}$ becomes unbounded while $\max_k |\mu_k|$ may very well remain fixed or go to zero. \end{observation}

Under a non-resonance condition, it is possible to obtain stronger results \cite{wu2025quantum_resonance}: given $\eps, T > 0$, the error between solution $y_N$ to the $N$-level Carleman system and the solution $u$ to the original system at time $T$ \alnn{
\norm{u(T) - \Tc_1(y_N(T))   } \leq \eps
}
for $N=O(\log(T/\eps)/ \log(1/R_r))$ for $R_r = 4e\mu\kappa_1^N\norm{W_2}_1 / R_\Delta$ where $\kappa_1^N$ is the condition number for the matrix that diagonalizes $W_1$, $\norm{u(t)}_{t\leq T} \leq \mu$ and $R_\Delta$ a resonance parameter related to the spectrum of $W_2$. The question of convergence is still not resolved as the situation below illustrates.

\begin{observation}\label{obv_need_consistency}
When $W_i$'s are unbounded operators with discretizations $W^N_i$, so $\norm{W^N_i}$ grows without bound. Then for a fixed $\eps$, the number of Carleman levels needed $N:=N(\eps)$ goes to infinity. To get at a solution to the original system, it must be established that limit of the solution $y_N$ to the discretized systems, $y:=\lim_{N\to \infty}y_N$, exists, and in the vanishing $\eps$-limit, $y$  must additionally be a solution to the infinite level Carleman system over the infinite dimensional Hilbert space whose solutions are unique, and therefore has some notion of well-posedness.
\end{observation}

Even if one were to directly substitute the unbounded operators in a finite truncation of the Carleman system, then convergence and well-posedness are non-obvious; the full Carleman system requires more. The key point here is that in analogy with the Lax-Richtmeyer theorem for linear PDEs one needs the additional notion of consistency of the discretizations to guarantee that the solutions converge. The Trotter-Kato approximation theorem easily offers a resolution to the Observation~\ref{obv_need_consistency} as we will detail in Section~\ref{linear_semigroup_theory}). We show that because of the recursive structure, the convergence of Carleman linearization for  equation~\eqref{eq_orig_nonlinear_eqn}, with self-adjoint $W_1$, can be related to the dissipativity of the linear system \alnn{\label{eq_2level_0subdiagonal_system}
\phi' = W\phi,\quad W:H\oplus H^{\tsr 2}\to H\oplus H^{\tsr 2}, \txt{ where } W := \pmat{W_1 & W_2 \\ 0 &\smfrac{1}{2}\SymmSum_2(W_1)},
}
and the nonlinearity parameter $R$ from Observation~\ref{obv_r_bound} is a direct control on the dissipativity.

The results are obtained by applying the linear semigroup theory to the infinite dimensional Carleman operator $\CarlemanOp{W_1,W_2}$. The convergence results are general and do not appeal to resonance. At the same time, the question of convergence rate is reduced to convergence rates for the Trotter-Kato approximation theorem, which has been well studied. Careful, albeit problem-specific, analysis can be used to deal with unbounded $W_i$'s as we will illustrate with an example.

The semigroup methods also extend the convergence to non-dissipative dynamics by applying perturbations to the Carleman operator generated semigroup. In line, with results showing that systems with positive Lyapunov exponent do not have efficient quantum algorithms \cite{lewis2024limitations}, the dissipativity of $W_1$ is necessary for $\CarlemanOp{W_1, W_2}$ to be a strongly continuous  semigroup generator. However, perturbations $W_2$ for which solution norm can grow corresponding to spectrum of $\CarlemanOp{W_1, W_2}$ being bounded above can also be treated.

\subsection{Basic motivation}\label{section:basic_motivation}

Our approach starts by giving the space $\Fock(H)$, on which the Carleman linearization lives a Hilbert space structure using an inner-product $Q$. We first consider initial conditions, $\phi_0$, from the unit ball, $\norm{\phi_0}<1$, and show that the initial condition lifted $\Fock(H), \FExp(\phi_0)$ is in $L^2(\Fock(H), Q)$. We will then use operator semigroup theory with respect to the Hilbertian structure on $\Fock(H)$.

The solution $\hat u \in  L^2(\Fock(H), Q)$ to the full infinite-dimensional system (equation~\eqref{eq_full_carleman_system}) is given by the semigroup $T_t$ generated by $\CarlemanOp{W_1,W_2}$ evaluated on the initial condition, $\hat u(0) = \FExp(\phi_0)$ if the problem is well-posed. Projecting onto $H$ via $\Tc_1$ gives the solution to the original system. While the full infinite-dimensional semigroup necessitates working with resolvents, under some conditions it can be approximated by the finite truncations $\CarlemanOp{W_1,W_2}_N$ and their semigroups $T_{t, N}$ which are just the usual operator exponentials.
Since $\Tc_1 T_t \FExp(\phi_0) = u(t)$, this implies that the convergence rate of the error, $e(N)=\norm{\Tc_1 T_{t, N} \oplus_{n\leq N}\phi_0^{\tsr n} - u(t)}_H$, to zero is dominated by the convergence of the truncations\aln{
e(N)&=\norm{\Tc_1 T_{t, N} \FExp_N(\phi_0) - u(t)}_H \leq \norm{ T_{t, N} \FExp_N(\phi_0) - T_{t}\FExp(\phi_0)}, 
}
where $\FExp_N(v) := \oplus_{n\leq N}v ^{\tsr n}$.

This convergence of the error to zero in the large $N$-limit requires that the infinite system is well-posed and the generators are well behaved. In this article we will address this by examining the structure of the Carleman operator and using the Trotter-Kato approximation theorem. We will establish that the truncations converge when the simplified system  (equation~\eqref{eq_2level_0subdiagonal_system}) is dissipative and the infinite Carleman system is well-posed.

The convergence of the evolution semigroup to the solution of the defining PDE necessitates using $\norm{\phi_0}_H < 1$. However, this requirement can be dropped by adapting the equations. The dissipativity hypothesis for $W, W_1$ along with $\norm{\phi_0}_H < 1$ are much more widely applicable than the standard theory (including the recent work of \cite{jennings2025quantum}) as discussed in Section~\ref{section_convergence_comparison}; in particular, the discretization limit can be dealt with since the convergence condition does not involve operator norms.

The semigroup approach is motivated by trying to formalize the correspondence of the original nonlinear dynamics $\phi' = W_1\phi + W_2\phi^{\tsr 2}$ and the Carleman linearized dynamics: the regularity properties of the semigroup guarantee that the correspondence is well behaved and one can solve the linear problem to get at the nonlinear problem. Before delving into unbounded $W_i$'s and perturbations, we detail this correspondence. For this we will need to introduce some background.

\emph{Some notational remarks}: $[k]$ denotes the ordered set $\{1, 2\dots k\}$. For Hilbert space $H$, $\Bc(H)$ denotes bounded linear operators on $H$, and $\Lc(H)$ linear, possibly unbounded operators on $H$. When the inner-product $Q$ is clear from context, $L^2(\Fock(H), Q)$ and $\Fock(H)$ will be used interchangeably.

\subsection{Linear semigroup theory}\label{linear_semigroup_theory}

We start by recalling the definition of a $C_0$-semigroup.

\begin{definition}
  Let \( X \) be a Banach space. A family of bounded linear operators \( (T_t)_{t \geq 0} \subset \mathcal{B}(X) \) is called a \emph{\( C_0 \)-semigroup} (or strongly continuous semigroup) if the following holds:
\begin{enumerate}
    \item Semigroup property: $T_0 = I, \quad T_{t + s} = T_t\circ T_s$ for all $t, s \geq 0$.
    \item Continuity in strong operator topology: $\lim_{t \to 0^+}\norm{T_t x-x} = 0$ for all $x\in X$.
\end{enumerate}
\end{definition}

For bounded generators, the semigroups are uniformly continuous semigroups, $\lim_{t \to 0^+}\norm{T_t - \one} = 0$. Similar to  semigroups of linear operators, a nonlinear semigroup $(T_t)_{t\geq 0}$ is a semigroup with $T_t$ possibly nonlinear. For nonlinear semigroups, knowing that one exists does not mean one knows an exact closed form. Unless explicitly qualified, the semigroups will be assumed to be linear.

Associated to a $C_0$-semigroup $T_t$ on a Banach space $X$ are operators $A$, the generator, and $R(\lambda, A):=(\lambda \one-A)^{-1}$, the resolvent, parametrized by $\lambda\in\rho(A):=\{\lambda\in \CC:\lambda\one-A\txt{ is invertible }\}$. $A$ can be recovered from the semigroup by$$
Ax = \lim_{h\to 0^+}\dfrac{T_hx - x}{h}\ ,
$$
while $R(\lambda, A)$ is related to the semigroup via the Laplace transform $$
R(\lambda, A)x = \lim_{t\to \infty}\int_0^t e^{-\lambda s}T_s(x)ds .
$$
When the generator $A$ is bounded then for $t\geq 0$, $T_t = \exp(tA) = \sum_{k=0}^{\infty} t^kA^k/k!$, while for unbounded $A$, generally one needs to use the resolvent to reconstruct the semigroup. The solution to a well-posed evolution equation $u'(t) = Au(t), t \geq 0, u(0) = u_0$, is given by the semigroup action, $u(t) = T_t(u_0)$ for each given initial value $u_0$. While a generator may be unbounded, it must be closable, that is, there must exist a closed extension.

\begin{definition}Let $X$ and $Y$ be Banach spaces, and let $T :X \to Y$ be a linear operator with domain $\Dom(T)$, then $T$ is defined to be \emph{closed} if the graph $
\mathcal{G}(T) = \left\{ (x, Tx) \in X \times Y : x \in \Dom(T) \right\}
$
is closed in the product topology on $X \times Y$. Equivalently, $T$ is closed if for every sequence $(x_n) \subset \Dom(T)$ such that
$
x_n \to x \in X \quad \text{and} \quad Tx_n \to y \in Y,
$
it follows that $x \in \Dom(T)$ and $Tx = y$. For any linear operator $A$, we will denote by $\closure{A}$ a closed extension of $A$.
\end{definition}

We note that the semigroup is not a numerical scheme since the parameter $t$ is continuous, although numerical schemes can be derived by discretizing $t$. We will need the following property of closability and domain of an operator, and a notion of an appropriate dense subspace of the domain of an operator.

\begin{property}\label{property_closability}(See, for instance, \cite{schmudgen_unbounded_operators})
A linear operator $A:\Dom(A) \to H$ is closable if and only if the adjoint $A^*$ is densely defined.  Additionally, note that $\Dom(A^*)$ can be understood using that $\phi \in \Dom(A^*)$ iff $|\inner{A\psi, \phi}|\leq C_\phi\norm{\psi}$ for all $\psi\in \Dom(A)$ where $C_\phi$ depends only on $\phi$.
\end{property}

\begin{definition} A subspace $D\subset X$ is core for a linear operator $A:\Dom(A)\subset X \to X$ if $D$ is dense in $\Dom(A)$ in the graph norm $\norm{x}_A := \norm{x} + \norm{Ax}$. 
\end{definition}

The (first) Trotter-Kato approximation theorem relates the convergence of generators to the convergence of their semigroups.

\begin{statement}(The (first) Trotter-Kato approximation theorem, see \cite{engel_nagel_semigroups}) Suppose $T(t),T_n(t)$ with $t\geq 0$, $n \in \NN$, are strongly continuous
semigroups on Banach space $X$ with generators $A$ and $A_n$, respectively, and assume that
they satisfy the estimate $\norm{T(t)}, \norm{T_n(t)}\leq M\exp(\omega t)$ for all $t\geq 0, n\in\NN$ with appropriate constants $\omega\in \RR, M\geq 1$.
Suppose $D$ is a core for $A$, 
then the following are equivalent

\begin{enumerate}
    \item $T_n(t)x\to T(t)x$ for all $x \in X$, uniformly for $t$ on compact intervals.
    \item For each $x\in D$, there exists $x_n\in \Dom(A_n)$ such that $x_n\to x$ and $A_nx_n \to Ax$.
\end{enumerate}
\end{statement}

Contractive semigroups when $X$ is a Hilbert space are precisely those that are generated by densely defined dissipative operators.

\begin{statement}(Lumer–Phillips Theorem, see \cite[Theorem~II.3.15]{engel_nagel_semigroups}) 
Let \( X \) be a Banach space, and let \( A : \Dom(A) \subseteq X \to X \) be a densely defined linear operator. Then \( \closure{A} \) is the generator of a strongly continuous semigroup \( (T(t))_{t \geq 0} \) of contractions on \( X \) if and only if:
\begin{enumerate}
    \item \( A \) is \emph{dissipative}, i.e., for all \( x \in \Dom(A) \) and all \( \lambda > 0 \)
    \[
    \|(\lambda I - A)x\| \geq \lambda \|x\|.
    \]
    \item The range \( \text{Ran}(\lambda I - A) \) is dense for some \( \lambda > 0 \).
\end{enumerate}
In this case, the semigroup \( (T(t))_{t \geq 0} \) satisfies \( \|T(t)\| \leq 1 \) for all \( t \geq 0 \).
\begin{remark}\label{rem_dissipativity_on_hilbert_spaces}  We have that if both $A$ and $A^*$ are dissipative, then $\closure{A}$ generates a contraction semigroup on $X$ \cite[Corollary~II.3.17]{engel_nagel_semigroups}. Additionally, on Hilbert spaces, dissipativity has a nice characterization (see \cite[Example~3.26, Proposition~II.3.23]{engel_nagel_semigroups}: 
when $X=H$
for Hilbert space $H$, then $A$ being dissipative is equivalent to the bound \alnn{
\txt{Re}\inner{Av, v} \leq 0 \txt{ for all } v\in \Dom(A).
}
And last, if $A$ is dissipative and $\Dom(A)$ is dense in $H$, then $A$ is closable and $\closure{A}$ is dissipative (see, for example, 
\cite[Theorem~1.4.5]{pazy2012semigroups}). 
\end{remark}
\end{statement}

More generally, to consider nondissipative dynamics, as in dynamics where solution norm may grow, we will need a characterization for generators of not necessarily contractive semigroups. If a semigroup $(T(t))$ satisfies a bound of form, $\norm{T(t)}, \norm{T_n(t)}\leq M\exp(\omega t), M\geq 1$, then the semigroup is said to be of type $(M, \omega)$. Semigroups of type $(1, \omega)$ are quasi-contractive, type $(1, 0)$ being contractive. The Feller-Miyadera-Phillips theorem characterizes the generators of $(M, \omega)$ $C_0$-semigroups.

\begin{statement}(Feller-Miyadera-Phillips theorem, see  \cite[Theorem~II.3.8]{engel2006short}) A linear operator $(A, \Dom(A))$ on a Banach space $X$ is the infinitesimal generator of a $(M, \omega)$ $C_0$-semigroup $T(t)$ if and only if\begin{enumerate}
    \item $A$ is densely defined and closed.
    \item For every $\lambda > \omega$, the operator $(\lambda I - A)$ is invertible 
    with the resolvent $R(\lambda, A) = (\lambda I - A)^{-1}$ satisfying for all $n\in \NN$, $\|R(\lambda, A)^n\| \le {M}/{(\lambda - \omega)^n}.$
\end{enumerate}
\end{statement}

\subsection{The evolution semigroup perspective}\label{sec_evolution_semigroup_perspective}

A correspondence between the possibly nonlinear dynamics on $H$ and the linear dynamics on $\Fock(H)$ is only possible if solutions to the lift of the original system \eqref{eq_orig_nonlinear_eqn} via Carleman linearization to $\Fock(H)$ exist and are unique. For Cauchy problems on Hilbert spaces, we have the following characterization of well-posedness in terms of regularity of the semigroup.

\begin{property}\label{prop_nonlinear_to_abstract_cacuhy_problem} Consider equation~\eqref{eq_orig_nonlinear_eqn}, $\phi' = W_1\phi + W_2\phi^{\tsr 2}, \phi(0)=\phi_0$, and let $A:=\CarlemanOp{W_1,W_2}$ be the associated Carleman operator on $\FreeFockOld(H) = \oplus_n H^{\tsr n}$ completed with respect to inner-product $\inner{\doubledot}_Q := \sum_k \inner{\doubledot}_{H^{\tsr k}}$, then the following hold:
\begin{enumerate}\label{eq_abs_cauchy_problem}
    \item The abstract Cauchy problem on $\Fock(H)$, $\psi'(t) = A\psi(t), \psi(0)=\psi_0\in \Dom A\subset \Fock(H),$ having a unique solution is equivalent to the operator $A$ generating a $C_0$-semigroup on $\Fock(H)$ (see, for instance, \cite[Chapter~II,Theorem~6.7]{engel_nagel_semigroups}).
\item If $u(t)$ is a solution to equation~\eqref{eq_orig_nonlinear_eqn}, then $\FExp u(t)$ is a solution to the abstract Cauchy problem on $\Fock(H)$ with $\psi_0 = \FExp \phi_0, A=\CarlemanOp{W_1,W_2}$ (see Section~\ref{section_carleman_semigroup}).
\end{enumerate}
Therefore, under the assumption that the Cauchy problem (equation~\eqref{eq_abs_cauchy_problem}) has a unique solution, the solution $u(t)$ to equation~\eqref{eq_orig_nonlinear_eqn} can be obtained from a solution $\psi(t)$ for equation~\eqref{eq_abs_cauchy_problem} by projecting onto the first tensor factor, giving a correspondence between the nonlinear evolution on $H$ under equation~\eqref{eq_orig_nonlinear_eqn} and linear evolution on $\Fock(H)$ generated by the Carleman operator.
\end{property}

\noindent Precisely, the semigroup generates the solution as follows.

\begin{statement}(ODEs on a Banach space, see \cite{pazy2012semigroups}) 
Consider the ordinary differential equation \alnn{\label{eq_abstract_cacuhy_problem_banach_space}
\psi'(t) = A\psi(t) + f(t), \ \psi(0)=\psi_0
}
for $f:[0, T)\to X$, $X$ a Banach space. Then if $A$ is the infinitesimal generator of a strongly continuous semigroup $T_t$,
the solution is given by $$\psi(t) = T_t\psi_0 + \int_0^t T_{t-s}f(s)ds ,\  t\geq 0.$$

In particular, when $\psi_0\in \Dom(A), f=0$, the semigroup action, $T_t\psi_0$, gives the classical solution when $A$ is the generator of a $C_0$-semigroup.
\end{statement}

For $f=0$, stability implies that $\norm{\psi(t)}\leq \norm{\psi_0}$, meaning that not only must the semigroup be a $C_0$-semigroup, it must also be contractive.  This is precisely the regime of the Lumer–Phillips Theorem. We will proceed by establishing conditions under which $\CarlemanOp{W_1,W_2}$ generates a $C_0$-semigroup so the uniqueness assumption holds. In fact, the semigroup will be contractive for dissipative equations. Using the Trotter-Kato approximation, the convergence of finite truncations of the Carleman system to the infinite limit will be proved, so the semigroup can be approximated by finite-dimensional matrix exponentials.

We will also be interested in \emph{mild} solutions. The solution $\psi$ to the abstract Cauchy problem is a mild solution if we have $$
\int_0^t \psi(s)ds \in \Dom(A), \txt{ } A \int_0^t \psi(s)ds = \psi(t)-\psi_0 +\int^t_0 f(s)ds. 
$$ The solution becomes a classical (\emph{strong}) solution if $u\in \C^1(\RR_+, X)$, that is, $\psi$ is differentiable.

The semigroup approach extends to the case where the coefficients are unbounded, and proves to be an essential ingredient in establishing that the solution to the ODE, defined by the semi-discretized PDE, converges to the solution of the PDE in the limit. The Carleman linearization as used here is a way to embed a nonlinear semigroup into a linear semigroup.

\subsection{Nondissipative perturbations}\label{section_perturbed_nondissipative_dynamics}
For a nonlinear operator $L$, by the Carleman semigroup, we will mean the semigroup generated by the Carleman linearization $\CarlemanOp{L}$. The connection to well-posedness theory implies that we need the semigroup to be a $C_0$-semigroup.
When $L=L_0 + \Nc$ for a linear operator $L_0$ and a nonlinear $\Nc$, the Carleman linearization corresponds $\CarlemanOp{L} = \CarlemanOp{L_0} + \CarlemanOp{\Nc}$, the perturbation theory of linear semigroups comes into play. For a large class of perturbations, especially ones where the dynamics may not be dissipative, the growth bounds on the resolvent of the type needed by the Lumer-Phillips and Feller-Miyadera-Phillips theorem are not met.

Since the semigroup is defined by the inverse Laplace transform of the resolvent, one only needs the Laplace transform to exist. The notion of well-posedness must be relaxed, but for a large class of perturbations $\CarlemanOp{\Nc}$ on initial conditions with sufficient regularity, it is still useful. The generated semigroups are the  $1$-integrated semigroups, and the well-posedness of the Cauchy problem is in terms of the mild, and not strong, solutions. Intuitively, a $1$-integrated semigroup, $S_t$, arises as an integral of a $C_0$ semigroup, $T_t$,  $ S_t = \int_0^t T_s ds$.
If a mild solution $u(t)$ exists to the abstract Cauchy problem, equation~\eqref{eq_abstract_cacuhy_problem_banach_space}, then for \aln{
v(t) = S_t\psi_0 + \int_0^t S_s f(t-s)ds,
}
$v\in \C^1(\RR_+, X)$ and $u$ can be obtained by differentiating $v$, $u=v'$. 
If $u$ is a classical solution then $v\in\C^2(\RR_+, X)$. We have that if $A$ is the generator of a once integrated semigroup on $X$ with $x \in \Dom (A^{k+1})$ then there exists a unique classical solution of \eqref{eq_abstract_cacuhy_problem_banach_space} (see \cite[corollary~3.2.11]{arendtvectorlaplcetransforms}).  We refer to \cite{arendtvectorlaplcetransforms} for background on integrated semigroups, and for our purposes we just need a characterization in terms of resolvent growth bounds: a linear operator is a generator for a $1$-integrated semigroup if there exists $\omega \geq 0, M \geq 0, b \in \NN$ such that for all $\lambda \in \rho(A), \Re(\lambda)>\omega$,  \alnn{\label{eq_resolvent_bound_once_integrated_semigroup}
\norm{R(\lambda, A)} \leq M|\lambda|^{-b}
}
(see \cite[Theorem~3.2.8]{arendtvectorlaplcetransforms}). With this resolvent bound we can handle nonlinear perturbation under Carleman linearization. As an example we will consider polynomial perturbations and obtain the convergence for the Carleman linearization of reaction-diffusion type equations $u' = \alpha \laplace u + u^{\tsr p},\  p\in \NN$, $\laplace$ being a Laplacian. Trotter-Kato type approximation theorems are known for once integrated semigroups, but we will not comment in detail on the question of computing and approximating integrated semigroups.

\section{The Carleman semigroup}\label{section_carleman_semigroup}

In this section we will work with dynamical systems on finite-dimensional Hilbert spaces. Consider \alnn{
\phi' = W_1\phi + W_2\phi^{\tsr 2},\label{eq_orig_quadratic_system}
} with $\phi(0)=\phi_0$ on the Hilbert space $H,\dim H = d < \infty$. Let $H\ni \phi(t) = (\phi_i(t))_{i\in [d]}$ with respect to some basis $(h_i)$ for $H$, that is, $\phi_i(t) = \inner{\phi(t), h_i}$, then $\phi^{\tsr 2}(t) = (\phi_i(t)\phi_j(t))_{i, j\in [d]}$  on $H^{\tsr 2}$ with basis $(h_i\tsr h_j)_{i,j\in [d]}$. This yields $
(\phi^{\tsr 2})' = (\phi_i\phi_j)'_{ij} = (\phi'_i\phi_j + \phi_i\phi'_j)_{ij} = \phi \tsr \phi' + \phi' \tsr \phi
$. For any operator $A$, define the operator $\SymmSum_n(A)$ mapping into $\Lc(H^{\tsr n})$ by \aln{
\SymmSum_n(A) &= \sum_{i=0}^{n-1}  \one^{\tsr i} \tsr A \tsr \one^{\tsr n-i-1} \\ &= A\tsr \one^{n-1} + ... \one^{j-1}\tsr A\tsr\one^{n-j-1} ... + \one^{n-1}\tsr A
:= \sum_{j\in[1:n]}\SymmSum^j_n(A),
}
with the convention that $\one^{\tsr 0} = 1\in \RR$, making $\SymmSum_1(A)=A$, and in $\SymmSum^j_n(A)$ $j$ indicates position of $A$. Notice that $\SymmSum_k$ is linear and commutes with the adjoint. Consider $D:H\to H, D\phi = \phi' = W_1\phi + W_2\phi^{\tsr 2}$. Combining $\phi' = W_1\phi + W_2\phi^{\tsr 2} $ with the infinite family of equations for $\phi^{\tsr n}, n\in \NN$, then expanding $(\phi^{\tsr 2})' =  \phi \tsr D\phi +   D\phi \tsr \phi$, we have  \alnn{
(\phi^{\tsr 2})'
&=  \phi \tsr\left[ W_1\phi + W_2\phi^{\tsr 2}\right] +  \left[ W_1\phi + W_2\phi^{\tsr 2}\right] \tsr \phi = \SymmSum_2(D)\phi^{\tsr 2},\nonumber \\
(\phi^{\tsr n})' &= \sum_{i=1}^n  \phi^{\tsr i-1} \tsr (D\phi) \tsr \phi^{\tsr n-i-1} = \SymmSum_n(D)\phi^{\tsr n},\label{eq_tensor_power_derivative}}
yields the infinite-dimensional Carleman system, \alnn{
\Diff{t}\pmat{\phi\\ \phi^{\tsr 2}\\ \vdots \\ \phi^{\tsr n}\\ \vdots} &= \pmat{
\SymmSum_1(W_1) &\SymmSum_1(W_2) &0 &0 &0 &\dots\\
0 &\SymmSum_2(W_1) &\SymmSum_2(W_2) &0 &0 &\dots \\
\vdots &\vdots\\
0 &\dots &\SymmSum_n(W_1) &\SymmSum_n(W_2) &0 &\dots \\
\vdots &\vdots\\
} \pmat{\phi\\ \phi^{\tsr 2}\\ \vdots \\ \phi^{\tsr n}\\ \vdots}\nonumber \\
&:= \CarlemanOp{W_1,W_2}\FExp{\phi}\label{eq_infinite_carleman_expansion}
}
with finite sections that are the truncated system at Carleman level $N$ corresponding to the $K(N)$-dimensional dynamics with $K(N)=\sum_{k=1}^N d^k$.  \alnn{
\Diff{t}\pmat{\phi\\ \phi^{\tsr 2}\\ \vdots \\ \phi^{\tsr N}} &= \pmat{
\SymmSum_1(W_1) &\SymmSum_1(W_2)  &0 &\dots\\
0 &\SymmSum_2(W_1) &\SymmSum_2(W_2)  &\dots \\
\vdots &\vdots\\
0 &\dots &\dots  &\SymmSum_N(W_1)}_{K(N)\times K(N)} \pmat{\phi\\ \phi^{\tsr 2}\\ \vdots \\ \phi^{\tsr N}}_{K(N)\times 1}\\
&:= \CarlemanOp{W_1,W_2}_N \FExp_N{\phi},\label{eq_finite_carleman_expansion}
}
where $\FExp:H\to \Fock(H), \FExp(\phi) = \oplus_n \phi^{\tsr n}$ and $\FExp_N(\phi) = \oplus_{n\leq N} \phi^{\tsr n}$.
The dynamics  for $\CarlemanOp{W_1,W_2}$ live on the space $\FreeFockOld(H) = \oplus_{\ZZnn} H^{\tsr n}$, while the dynamics generated by the truncated operator $\CarlemanOp{W_1,W_2}_N$ live on $\Fock_N(H) := \oplus_{n\leq N} H^{\tsr n}$. Note that $\CC := H^{\tsr 0}$ decouples from the dynamics on $\oplus_{\NN} H^{\tsr n}$.

\begin{observation} By construction if $u(t)$ is a solution to $\phi' = W_1\phi + W_2\phi^{\tsr 2}$ with initial condition $\phi_0, \norm{\phi_0}< 1$, then $\FExp(u(t))$ satisfies $$\FExp(u(t))' = \CarlemanOp{W_1,W_2}\FExp(u(t)),\ \FExp(u(0)) = \FExp(\phi_0).
$$
The nonlinear dynamics, therefore, can be identified with an abstract Cauchy problen on $\Fock(H)$ as noted in property~\ref{prop_nonlinear_to_abstract_cacuhy_problem} if the Carleman linearized system has unique solutions.
\end{observation}

For any Hilbert space $H$, the free Fock space over $H$, $\FreeFockOld(H)$, is the vector-space $\oplus_{n\in \NN_0} H^{\tsr n}$ with $H^0 = \CC$. Each $v\in \FreeFockOld(H)$ is the vector $(v_k)_{\ZZnn}$ with $v_k\in H^{\tsr k}$. For the dynamics considered here, $H^0$, plays no role. Without loss of generality, the $H^0$ factor is dropped. We need to set the following notation:

\begin{itemize}

\item The finite $N$-particle spaces are the finite direct sums $\Fock_N(H) = \oplus_{n\leq N} H^{\tsr n}$. $v\in \Fock_N(H)$ is the $N$-vector $(v_m)_{m\in [N]}, v_m\in H^{\tsr m}$.
\item The operators $\Tc_k$ are projections on $H^{\tsr k}$, $$
\Tc_k: \Fock_N(H), \Fock(H)\to H^{\tsr k},\ \Tc_k(\oplus_i v_i) = v_k$$ where for $\Tc_k:\Fock_N(H)\to H^{\tsr k}, k\leq N$.
\item The operators $\Pc_k, \Jc_k$ define projections and inclusions of finite particle spaces through which $\Fock_k(H)$ are identified inside $\Fock(H)$
\alnn{
\Pc_k&: \Fock(H)\to  \Fock_k(H),\ \Pc_k(\oplus_i v_i) = \oplus_{i\leq k} v_i,\nonumber \\ \Jc_k&:\Fock_k(H)\to \Fock(H),\ \Jc_k(\oplus_{i\leq k} v_i) =  [\oplus_{i\leq k} v_i]\oplus [\oplus_{i>k} 0],\nonumber\\ &
\txt{ with } \Pc_k\Jc_k = \one|_{\Fock_k(H)}\label{eq_finite_fock_projections_inclusions}.
}

\end{itemize}

Now we would like to translate the results on convergence of the Carleman linearization in terms of semigroups on $\Fock(H)$ fleshing out the discussion in section~\ref{section:basic_motivation}.
From \cite{liu2021_dissipative_nonlinear_carleman} we have that if $\dim H = d < \infty$, $W_i$'s are bounded, $\lambda_1 = \max_{\txt{Re}(\txt{spec}(W_1))} < 0$, then using iterative methods on the $N$-level Carleman system, one can construct $\hat u_N\in H$ such that $\hat u_N$ approximates the solution, $u$, to the nonlinear system in equation~\eqref{eq_orig_quadratic_system}. Moreover, the error $\eta_j := u^{\tsr j} - \hat u_N$ satisfies $\norm{\eta_j} \leq \norm{\phi_0}^{j}R^{N+1-j}$ (see \cite[Corollary~1]{liu2021_dissipative_nonlinear_carleman}) for $j>1$, while for $\eta_1$ a stronger bound \alnn{
\norm{\eta_1} \leq \norm{\phi_0}R^N(1-\exp(\txt{Re}(\lambda_1)t))^N \label{eq_R_carleman_cconvergence}
}
holds, where $R=\norm{W_2}\norm{\phi_0}/|\txt{Re}(\lambda_1)|$. Therefore, for $R<1$, $\eta_1$ vanishes with $N$ large. Restating in terms of the solutions $\phi_N, \phi$ to the truncated  $N$-level system, equation~\eqref{eq_finite_carleman_expansion}, and the full Carleman system, equation~\eqref{eq_infinite_carleman_expansion}, the error is given by $e_j(N) = \Tc_j \phi_N  -\Tc_j\phi$. Now suppose $\CarlemanOp{W_1,W_2}_N,\ \CarlemanOp{W_1,W_2}$ generate $C_0$-semigroups $T_{t, N}, T_t$, then since $\Tc_1\phi = u$, and $\phi_N, \phi$ must be given by the semigroup action,\aln{
\norm{e_1(N)}&=\norm{\Tc_1 T_{t, N} \FExp_N(\phi_0) - u(t)}_H \leq \norm{ T_{t, N} \FExp_N(\phi_0) - T_{t}\FExp(\phi_0)}. 
}

The semigroup approach, therefore, reduces to the following questions of showing that $\CarlemanOp{W_1,W_2}_N$, $\CarlemanOp{W_1,W_2}$ are $C_0$-semigroup generators, and the truncated Carleman semigroups with generators, $\CarlemanOp{W_1,W_2}_N$,  converge to the Carleman semigroup generated by the untruncated $\CarlemanOp{W_1,W_2}$.
And we would like to answer the questions without dependence on $R$ or the norms $\norm{W_i}$'s.

\subsection{$L^2(\Fock(H), Q)$ Carleman linearization}

Consider the subspace $\Fock(H)^0:=\cup_k\Fock_k(H)$ inside $\Fock(H)$ consisting of vectors $v=(v_i)$ such that there exists $N_v$ with the property that $v_i=0$ for all $i>N_v$, meaning $v\in \Fock_{N_v}(H)$. Note that for all $u, v\in \Fock(H), \alpha, \beta\in \KK$, $\alpha u + \beta v\in \Fock^0(H)$ where $\KK=\RR, \CC$ is the field underlying  $H$.

\begin{observation}
We take the Hilbert space completion of $\Fock(H)^0$ with respect to the norm $\norm{\cdot}_Q$ associated with the following inner-product: \alnn{
\inner{\doubledot}_Q = \sum_{k} \inner{\doubledot}_{H^{\tsr k}}\label{eq_identity_fock_innerproduct}.
}
Then by construction, $(\Fock(H), Q)$ being the completion of $\Fock(H)^0$, $\Fock(H)^0$ is dense in $(\Fock(H), Q)$. Specifically, $v=(v_m)_{m\in \NN}\in (\Fock(H), Q)$ means $\sum_m \inner{v_m, v_m}_{H^{\tsr m}} <\infty$ and, therefore, $v$ can be approximated arbitrarily well by  elements from finite particle space. Using $\dim H^{\tsr k} = d^k$, the matrix representation for $\inner{\doubledot}_Q \big|_{\Fock_N(H)}$ is given by $\txt{Diag}(1,\dots,Q_k, \dots Q_N)$ where  $Q_k:=\one_{d^{k}\times d^{k}}$.
\end{observation}

\noindent We will consider the Carleman linearization (equations~\eqref{eq_infinite_carleman_expansion}, \eqref{eq_finite_carleman_expansion}) on the Hilbert space $(\Fock(H), Q)$.

\begin{prop} Suppose $\phi\in H, \norm{\phi}<1$. Then $\FExp(\phi)\in \Dom(\CarlemanOp{W_1,W_2})$.
\end{prop}
\begin{proof} Using $\sum_k |x|^k$ converges uniformly for $|x|<1$, $\FExp(\phi)\in L^2(\Fock(H))$ since $
\norm{\FExp(\phi)} \leq  \sum_k \norm{\phi^{\tsr k}}\leq \sum_k \norm{\phi}^k < \infty$.
Now $\Tc_k (\CarlemanOp{W_1,W_2}\FExp(\phi)) = \Sc_k(W_1)\phi^{\tsr k} + \Sc_k(W_2)\phi^{\tsr k + 1}$.  Therefore,
\aln{
\norm{\Tc_k (\CarlemanOp{W_1,W_2}\FExp(\phi)} &\leq k\norm{W_1}\norm{\phi}^{\tsr k}  + k\norm{W_2}\norm{\phi}^{\tsr k+1}\\
&\leq \norm{\phi} \left[\norm{W_1}+\norm{W_2}\right] k \norm{\phi}^{k-1},
}
where we used $\norm{\phi}<1$. Since the uniform convergence of $\sum_k |\phi|^k$ for $|\phi|<1$ implies that term-by-term differentiated series $\sum_k k|\phi|^{k-1}$ also converges, we have \aln{
\norm{\CarlemanOp{W_1,W_2}\FExp(\phi)} &\leq  \norm{\phi} \left[\norm{W_1}+\norm{W_2}\right] \sum_k k \norm{\phi}^{k-1} < \infty.
}
\end{proof}

\begin{remark}\label{rem_adapted_system}
In analogy with the usual boson Fock space, one ideally wants to consider $\inner{\doubledot}_Q = \sum_{k} k!\inner{\doubledot}_{H^{\tsr k}}$. However, this dependence on $k$ gets in the way when trying to reduce dissipativity for $\CarlemanOp{W_1, W_2}$ to that of a finite level $\CarlemanOp{W_1, W_2}_N$. At the same time, we note that it can simplify checking analytic properties, and can prove very useful. The exponential map $\FExp:H\to \Fock(H)$ is only defined on the interior of the unit ball unlike on the usual Fock space (see, for example, \cite{meyer1995quantum}). However, the constraint $\norm{\phi_0}  < 1$ can be removed by a non-dimensionalization type rescaling. To see this, note that for any $M > \norm{\phi_0}$, by the multilinearity of $W_2$, $\psi = \phi/M$ satisfies
\alnn{\label{eq_adapted_system}
M\psi'= M W_1\psi + M^2 W_2\psi^{\tsr 2}, \psi(0) = \phi_0/M
}
with $\norm{\psi(0)} < 1$, and now one analyzes this adapted system where $W_2$ is replaced by $MW_2$.
\end{remark}

Now because $H$ is finite-dimensional, $W_i$'s are bounded. Although the Carleman operator $\CarlemanOp{W_1,W_2}$ may be unbounded, it is closeable as established next. We will identify $\CarlemanOp{W_1,W_2}$ with its closure $\closure{\CarlemanOp{W_1,W_2}}$.

\begin{prop}\label{prop_carleman_closable}  For bounded operators, $W_i$, $\CarlemanOp{W_1,W_2}$ is densely defined and closable on Hilbert space $(\Fock(H), Q)$.
\end{prop}
\begin{proof} Since $(\Fock(H),Q)$ is the closure of $\Fock(H)^0:=\cup_N \Fock_N(H)$, and on each $\Fock_N(H)$, $\CarlemanOp{W_1,W_2}=\CarlemanOp{W_1,W_2}_N$ is a bounded operator, $\CarlemanOp{W_1,W_2}$ is densely defined.  We will use property~\ref{property_closability} to verify closability. Using $v \in \Dom(A^*)$ iff $|\inner{Au, v}|\leq C_v\norm{u}$ for all $u\in \Dom(A)$, it is enough to get a bound of this type for each $v\in \Fock(H)^0$.
For $u\in \Dom(\CarlemanOp{W_1,W_2})$ and any $v\in \Fock_N(H)$, $i>N$ means $\inner{\Tc_i\CarlemanOp{W_1,W_2}u, v_i}=0$, therefore, \aln{
\inner{\CarlemanOp{W_1,W_2}u, v} &= 
\inner{\CarlemanOp{W_1,W_2}_N u, v }_{\Fock} +  \inner{\SymmSum_N(W_2)u, \Tc_N v }_{H^{\tsr N}} \\
&\leq \norm{\CarlemanOp{W_1,W_2}_N + \SymmSum_N(W_2)}_N\norm{v} \norm{u}.
}
So the choice $C_v:=\norm{\CarlemanOp{W_1,W_2}_N + \SymmSum_N(W_2)}_N\norm{v}$ works. Finally, since $\Fock(H)^0$ is dense, $\CarlemanOp{W_1,W_2}^*$ is densely defined, making $\CarlemanOp{W_1,W_2}$ closable.
\end{proof}

\subsection{Fock space dissipativity for $\CarlemanOp{W_1,W_2}_N$}

To show dissipativity of the operator $\CarlemanOp{W_1,W_2}$ on the Fock space, $(\Fock(H), Q)$, the idea is to write $\CarlemanOp{W_1,W_2}$ as a sum of operators whose dissipativity can be established by the recursive structure of the Carleman linearization. This will involve adjoints with respect to the $Q$-inner-product restricted to finite particle spaces. While $Q$ is taken to be trivial when dealing with initial conditions $\phi_0$ in the unit ball $\norm{\phi_0}<1$, it will be useful to establish the following connection with the usual inner-product on the Euclidean space, $\inner{\doubledot}_{E}$, for general $Q$. We note the following elementary fact.

\begin{property}

For any diagonal positive definite matrix, $P\in \RR^{N\times N}$, defining the inner-product on $\CC^N$, $\inner{u, v}_P = \bar u^T P v$, the adjoint $A^*$ for any operator $A$ (in matrix representation with respect to Euclidean basis) can be computed
as $A^{*_P} = P^{-1}\bar A^TP $

where we use $A^{*_P}$ to denote adjoint with respect to the $P$ inner-product explicitly, $\bar A^T$ being the usual adjoint $A^\dagger$ on $\CC^N$ with respect to Euclidean inner-product. Therefore, positivity of $A+A^*$ with respect to $P$-inner-product reduces to positivity with respect to $\inner{\doubledot}_E$ since \aln{
\inner{u, (A + A^*)v}_P = \bar u^T P(P^{-1}A^\dagger P + A) v = \bar u^T(A^\dagger P + PA)v = \inner{u, (A^\dagger P + PA)v}_E .
}
\end{property}

\noindent We will only be interested in real valued inner-products $P:= \txt{Diag}(\dots Q_k\dots Q_N)$ where $Q_k$'s are simply identity.
Next the following computations help illustrate $\CarlemanOp{W_1,W_2}$'s recursive structure.

\begin{prop}\label{prop_positivity_under_left_identity_tensor} For Hilbert spaces $U, V, W$, suppose $A\in \Bc(U\oplus V)$ given by the block matrix \aln{
A = \pmat{A_{11} & A_{12}\\ A_{21} & A_{22} } \in \Bc(U\oplus V),} where $A_{ij}\in \Bc(U_i, U_j)$ with $U_1 = U, U_2 = V$, is such that $A\leq 0$, that is, for all $u\in U, v\in V$, \aln{
&\pmat{ u \\ v }^*\pmat{A_{11} & A_{12}\\ A_{21} & A_{22} }\pmat{ u \\ v } = u^* A_{11}u   +  u^* A_{12}v + v^* A_{21}u  +   v^* A_{22}v \leq 0, \\
&\hspace{2mm}\txt{where } \pmat{ u \\ v }^* = \pmat{ u^* \\ v^* }^T,
}
and ${}^*$ is the Hilbert space adjoint, then the operator $A_W\in \Bc(W\tsr U\oplus W\tsr V)$, \aln{A_W:=
\pmat{\one \tsr A_{11} &  \one \tsr A_{12}\\ \one \tsr A_{21} & \one \tsr A_{22} },\ A_W \leq 0.
}
\end{prop}
\begin{proof} Fix an orthonormal basis $(w_i)$ for $W$. We need to check that for any $w \in W\tsr U \oplus W\tsr V$, \aln{
w:= \pmat{ \sum_i w_i u_i\\ \sum_j w_i v_j } 
= \sum_i \pmat{ w_i u_i \\ w_i v_i }, \ w^* A_W w \leq 0
} holds. We have  \aln{
 w^* A_W w &= \sum_{ij} \pmat{ w_i \tsr u_i \\ w_i\tsr v_i }^*  \pmat{\one \tsr A_{11} & \one \tsr A_{12}\\ \one \tsr A_{21} & \one \tsr A_{22} } \pmat{ w_j \tsr u_j \\ w_j \tsr v_j } \\
 &=  \sum_{ij} \pmat{ w_i \tsr u_i \\ w_i\tsr v_i }^*  \pmat{w_j\tsr A_{11}u_j + w_j\tsr A_{12}v_j\\ w_j\tsr A_{21}u_j + w_j\tsr A_{22} v_j } \\
 &= \sum_{ij} \inner{w_i, w_j} \left( u_i^* A_{11}u_j   +  u_i^* A_{12}v_j + v_i^* A_{21}u_j  +   v_i^* A_{22}v_j \right)\\
 &= \sum_{i} \left( u_i^* A_{11}u_i   +  u_i^* A_{12}v_i + v_i^* A_{21}u_i  +   v_i^* A_{22}v_i \right)
 \leq 0,
}
where we used $ \inner{w_i, w_j}=\delta_{ij}$ by orthonormality of $w_i$'s and that $A\leq 0$ for the last inequality.
\end{proof}

\begin{lemma}\label{lemma_dissipativity_at_level0}  For any operator $A$ as in Proposition~\ref{prop_positivity_under_left_identity_tensor}, $A\leq 0$ implies \aln{
\omega_k:=\pmat{\SymmSum_k(A_{11}) & \SymmSum_k(A_{12})\\ \SymmSum_k(A_{21}) & \SymmSum_k(A_{22})} \leq 0 .
}
\end{lemma}

\begin{proof} We have \aln{
\omega_k &= \pmat{
\sum_i \SymmSum_k^i(A_{11}) &\sum_j \SymmSum_k^j(A_{12})  \\
  \sum_l\SymmSum_k^l(A_{21}) & \sum_m\SymmSum_k^m(A_{22})
} = \sum_{i\in [1:k]} \pmat{
\SymmSum_k^i(A_{11}) &\SymmSum_k^i(A_{12})  \\
  \SymmSum_k^i(A_{21}) & \SymmSum_k^i(A_{22})
} \\
&= \sum_i \pmat{ I_{i,k} \tsr A_{11} \tsr \one^{\tsr i} & I_{i,k}  \tsr A_{12} \tsr \one^{\tsr i}\\ I_{i,k}  \tsr A_{21} \tsr \one^{\tsr i} & I_{i,k}  \tsr A_{22} \tsr \one^{\tsr i} } \text{ with } I_{i,k} =\one^{\tsr(k-i-1)} \\ ,
&= \sum_i \pmat{ I_{i,k} \tsr A_{11}  & I_{i,k}  \tsr A_{12} \\
I_{i,k}  \tsr A_{21}  & I_{i,k}  \tsr A_{22}
} \tsr \one^{\tsr i} := \sum_i \omega^i_k\tsr \one^{\tsr i},
}
where to factor out $\one^{\tsr i}$ we used properties of Kronecker products.
Then each $\omega_k^i \leq 0$ by Proposition~ \ref{prop_positivity_under_left_identity_tensor} and the claim follows.
\end{proof}

\begin{remark} Lemma~\ref{lemma_dissipativity_at_level0} holds when the $A_{ij}$ are bounded operators on a possibly infinite-dimensional space. If $A_{ij}$ are unbounded, then closability is needed.
\end{remark}

Consider the truncated Carleman operator $\CarlemanOp{W_1,W_2}_N$. We will work with respect to a basis for $H$, with ${}^*$ denoting the abstract Hilbert space adjoint, either the standard Euclidean or the adjoint with respect to the matrix for the inner-product. Then $\CarlemanOp{W_1,W_2}_N$ is bounded, with finite dimensional support on which the inner-product is defined by a restriction of $Q$ with matrix representation, $P_{:N}=\oplus_{i\leq N} Q_N$. For dissipativity, we need the bound $$
\Re \inner{\CarlemanOp{W_1,W_2}_N v, v} = \smfrac{1}{2}\left(\inner{\CarlemanOp{W_1,W_2}_N v, v} +  \inner{\CarlemanOp{W_1,W_2}_N^* v, v}\right)\leq 0.
$$

\begin{observation}The Carleman operator $\CarlemanOp{W_1,W_2}_N$ can be decomposed into blocks as follows: \alnn{
&\pmat{
 \cviolet{\frac{1}{2}\SymmSum_1(W_1)} +  \frac{1}{2}\SymmSum_1(W_1) &\SymmSum_1(W_2) &\dots\\
0 & \frac{1}{2}\SymmSum_2(W_1)+\cred{\frac{1}{2}\SymmSum_2(W_1)} &\cred{\SymmSum_2(W_2)} &\dots  \\
0 &0 &\cred{\frac{1}{2}\SymmSum_3(W_1)} + \dots\\
0 &\dots &\dots  \\
0 &\dots &0 &\cblue{
\dots
} \frac{1}{2}\SymmSum_N(W_1),
}, \nonumber\\
 &\hspace{0.5in}\CarlemanOp{W_1,W_2}_N =  \frac{1}{2}\SymmSum_1(W_1) + \frac{1}{2}\SymmSum_N(W_1) + \sum_{k=1}^{N-1} \pmat{
\ddots &0 & \\
0 &Z_k  &0\\
 &0 &\ddots
},\label{eq_tiling_lower_diag_non_zero} \\
\nonumber
&\hspace{1in}\text{where } Z_k := {\pmat{\frac{1}{2}\SymmSum_k(W_1) &\SymmSum_k(W_2) \\ 0 &\frac{1}{2}\SymmSum_{k+1}(W_1)}_{d^k+d^{k+1}\times d^k+d^{k+1}}},
}
and for the projection $\Tc_k + \Tc_{k+1}:\Fock(H)\to H_k \oplus H_{k+1}$ and corresponding inclusions by zero-padding $\Ic_{k},\Ic_{k:k+1}:H^{\tsr k}, H^{\tsr k}\oplus H^{\tsr k+1} \to \Fock(H)$, $Z_k$ has been identified with $\Ic_{k, k+1}\circ Z_k\circ (\Tc_{k} + \Tc_{k+1})$ and $\SymmSum_k(W_1)$ with $\Ic_k \circ \SymmSum_k(W_1)\circ \Tc_k$.
We will analyze the positivity of each term in the sum\aln{ \CarlemanOp{W_1, W_2}_N + \CarlemanOp{W_1, W_2}_N^* = \SymmSum_1(W_S) + \SymmSum_N(W_S) + \sum_k\left(Z_k + Z_k^*\right),
}
with $W_S = \smfrac{1}{2}\left(W_1^*+W_1\right).$
Writing $Z_k^*$ with respect to the inner-product $P_{k:k+1}$ (which is identity since $c(n)=1$) yields \alnn{
\tau_k(u):&=\inner{u, \left[Z_k + P^{-1}_{k:k+1} Z^*_k P_{k:k+1}\right]u}_{P_{k:k+1}}= u^* \left[Z_k^* P_{k:k+1} + P_{k:k+1} Z_k \right]u \nonumber \\
\tau_k(u)&= \pmat{u_k \\ u_{k+1}}^* \left[  \pmat{
\frac{1}{2}\SymmSum_k(W_1^*) &0\\
 \SymmSum_k(W_2^*) &\frac{1}{2}\SymmSum_{k+1}(W_1^*) \\
} \right. \nonumber \\
&\hspace{10mm} \left. + \pmat{
\frac{1}{2}\SymmSum_k(W_1) &\SymmSum_k(W_2)\\
 0 &\frac{1}{2}\SymmSum_{k+1}(W_1) \\
} \right]\pmat{u_k \\ u_{k+1}}\nonumber \\
&= \left\langle u, \left[  \pmat{
\SymmSum_k\left(W_S\right) &\SymmSum_k(W_2)\\
 \SymmSum_k(W_2^*) &\SymmSum_{k+1}\left(W_S\right) \\
} \right]  u\right\rangle_E.\label{eq:non_vaninging_tiling}
}

\end{observation}
\noindent To apply lemma~\ref{lemma_dissipativity_at_level0} we need to write $\SymmSum_{k+1}$ in terms of $\SymmSum_{k}$ for which a recursive expansion is useful.

\begin{prop}\label{prop_symmsum_recursion}  For all $n$, $$  \SymmSum_{n+1}(A) = \frac{1}{2}\left( \SymmSum_n(\SymmSum_2(A)) + (\one^{\tsr n} \tsr A + A \tsr \one^{\tsr n}) \right) $$
with $\SymmSum_2(A)=\one \tsr A + A\tsr \one$.
\end{prop}
\begin{proof}
The following computation using $\SymmSum_n(X+Y) = \SymmSum_n(X) + \SymmSum_n(Y)$ and $\SymmSum_n(X) = \sum_{k=0}^{n-1} \one^{\tsr k} \tsr X \tsr \one^{\tsr n-1-k}$ verifies the claim, \aln{
\SymmSum_n(\one \tsr A) &+ \SymmSum_n(A\tsr \one) =  \sum_{k=0}^{n-1} \one^{\tsr k+1} \tsr A \tsr \one^{\tsr n-1-k} + \sum_{k=0}^{n-1} \one^{\tsr k} \tsr A \tsr \one^{\tsr n-k} \\
&= \sum_{k=1}^{n} \one^{\tsr k} \tsr A \tsr \one^{\tsr ((n+1)-1-k)} + \sum_{k=0}^{n-1} \one^{\tsr k} \tsr A \tsr \one^{\tsr ((n+1)-1-k)} \\
&= 2\cdot \sum_{k=0}^{n} \one^{\tsr k} \tsr A \tsr \one^{\tsr ((n+1)-1-k)} - \one^{\tsr n} \tsr A - A \tsr \one^{\tsr n} \\
&= 2\cdot \SymmSum_{n+1}(A) - (\one^{\tsr n} \tsr A + A \tsr \one^{\tsr n}).
}
\end{proof}

\begin{theorem}\label{theorem_dissipatinity_nonzero_tiling}  Suppose $W_S \leq 0$ and, in addition, \alnn{\label{eq_hypothesis_k}
\Lambda_1 := Z_1^* + Z_1 = \pmat{
W_S &W_2\\
 W_2^* &\frac{1}{2} \SymmSum_2(W_S) \\
} \leq 0.
}
Then $\CarlemanOp{W_1,W_2}^*_N + \CarlemanOp{W_1,W_2}_N  \leq 0$ for all $N$, so $\CarlemanOp{W_1,W_2}_N, \CarlemanOp{W_1,W_2}_N^*$ are dissipative.
\end{theorem}
\begin{proof}

First $W_S\leq 0$ implies $W_1+W_1^* \leq 0$ and it is immediate that each $\SymmSum_N(W_1^*+W_1) = \SymmSum_N(W_1^*)+\SymmSum_N(W_1) \leq 0$ because one can explicitly diagonalize $\SymmSum_N(W_1^*+W_1)$ with the eigenvalues are a positive scalar multiple of those for $W_1^*+W_1$. So from equation~\eqref{eq:non_vaninging_tiling}, it is enough to show for all $u\in H^{\tsr k}\oplus H^{\tsr (k+1)}$, each $\tau_k(u)\leq 0$ , \aln{
\tau_k(u)&=\langle u, \left[Z_k + P^{-1}_{k:k+1} Z^*_k P_{k:k+1}\right]u \rangle_{P_{k:k+1}}\\
&= \left\langle u, \left[  \pmat{
\SymmSum_k\left(W_S\right) &\SymmSum_k(W_2)\\
 \SymmSum_k(W_2^*) &\SymmSum_{k+1}\left(W_S\right) \\
} \right]  u\right\rangle_E \leq 0.
}
For each $k$, from equation~\eqref{eq_tiling_lower_diag_non_zero} and \eqref{eq:non_vaninging_tiling}, using Proposition~\ref{prop_symmsum_recursion}, it is equivalent to the bound \aln{
&\pmat{
\frac{1}{2}\SymmSum_k(W_1+W_1^*) &\SymmSum_k(W_2)\\
 \SymmSum_k(W_2^*) &\frac{1}{2}\SymmSum_{k+1}(W_1+W_1^*) \\
} \\ &= \pmat{
\SymmSum_k W_S &\SymmSum_k(W_2)\\
 \SymmSum_k(W_2^*) &\frac{1}{2}\SymmSum_{k}\SymmSum_2 W_S \\
} + \frac{1}{2}\pmat{0 &0 \\ 0 &  \one^{\tsr k} \tsr W_S + W_S \tsr \one^{\tsr k} }  \leq 0.
}
This follows since by lemma~\ref{lemma_dissipativity_at_level0} the bound in the hypothesis, equation~\eqref{eq_hypothesis_k}, implies the negative semi-definiteness of the first term, while $W_S\leq 0$ implies $\one^{\tsr k} \tsr W_S + W_S \tsr \one^{\tsr k}\leq 0$.

\end{proof}

\subsection{Convergence for bounded coefficients}\label{section_Convergence_for_bounded_coefficients}

Now we obtain convergence results for Carleman linearization in the limit of Carleman level when the Hilbert space is finite dimensional, $d:=\dim H<\infty$. Consider the system $
\phi' = W_1\phi + W_2\phi^{\tsr 2}, \phi(0) = \phi_0
$
with $W_1, W_2$ bounded. We are interested in convergence with respect to the norm induced by the inner-product on the space $\Hc:=(\Fock(H), Q)$. Now we will apply 
the Lumer-Phillips theorem to the abstract Cauchy problem on Fock space.
By remark~\ref{rem_dissipativity_on_hilbert_spaces}, we will use the dissipativity of the adjoint to avoid checking the range condition; additionally, the dissipativity on a dense subset implying the dissipativity of the closed extension will prove useful as it is enough to show dissipativity over finite particle space $\Fock(H)^0$ to get dissipativity on $\Fock(H)$. In fact,  $\Fock(H)^0$ form a core (see theorem~\ref{theorem_trotter_kato_applies}).
When $W_i$'s are unbounded, the dissipativity for $\CarlemanOp{W_1, W_2}^*$  will require a technical result, lemma~\ref{lemma_CWadjoint_dissipativity}, and it can be applied here as well.

\begin{prop}
Suppose $\CarlemanOp{W_1,W_2}$ is dissipative. Then so is $\CarlemanOp{W_1,W_2}^*$, and, therefore, it generates a contractive $C_0$-semigroup. Additionally, if $\CarlemanOp{W_1,W_2}_N$ is dissipative for each $N$, then $\CarlemanOp{W_1,W_2}$ is also dissipative.
\end{prop}
\begin{proof} We have that $ \Fock(H)^0\subset \Dom(\CarlemanOp{W_1,W_2})\cap \Dom(\CarlemanOp{W_1,W_2}^*)$, therefore, for all $v\in \Fock(H)^0$, \aln{
\txt{Re}\inner{\CarlemanOp{W_1,W_2}v, v} = \inner{(\CarlemanOp{W_1,W_2}+\CarlemanOp{W_1,W_2}^*)v, v}/2 = \txt{Re}\inner{\CarlemanOp{W_1,W_2}^*v, v} \leq 0  .
}
As $\CarlemanOp{W_1,W_2}, \CarlemanOp{W_1,W_2}^*$ are closed, with $\Fock(H)^0$ dense, this holds on the entire $\Dom(\CarlemanOp{W_1,W_2})$ and $ \Dom(\CarlemanOp{W_1,W_2}^*)$. For the second claim, restricted to each $\Fock_N(H)$, $\CarlemanOp{W_1,W_2}|_{\Fock_N(H)} = \CarlemanOp{W_1,W_2}_N$, so for any $v\in \Fock^0(H)$, $\txt{Re}\inner{\CarlemanOp{W_1,W_2} v, v} = \txt{Re}\inner{\CarlemanOp{W_1,W_2}_N v, v}   \leq 0$. And again as $\cup_N\Fock_N(H)$ is dense, $\CarlemanOp{W_1,W_2}$ closed, therefore, this holds over $\Dom(\CarlemanOp{W_1,W_2})$.
\end{proof}

Since $\CarlemanOp{W_1,W_2}$ is unbounded and non-normal, the semigroup $T_t$ generated by $\CarlemanOp{W_1,W_2}$ can only be defined in terms of the resolvent. However, we next show that it can be approximated by finite-dimensional semigroups with generators that are bounded, therefore, they are expressible as operator exponentials. For this we will appeal to the first Trotter-Kato approximation theorem.

\begin{theorem}\label{theorem_trotter_kato_applies}
Suppose $\CarlemanOp{W_1,W_2}$ and $\CarlemanOp{W_1,W_2}_N$ generate contractive semigroups $T(t), T_N(t)$, 
then for $t\in [0, C]$, $C$ finite, $$\lim_{N\to \infty } \norm{\Jc_N T_N(t)\Pc_Nx - T(t)x }= 0\ \text{uniformly}$$
for all $x\in \Fock(H)$, where $\Jc_N, \Pc_N$ are projections and inclusion maps defined in equation~\eqref{eq_finite_fock_projections_inclusions}.
\end{theorem}
\begin{proof} The point of using $\Jc_N T_N(t)\Pc_N$ is to identify the semigroups $T_N(t)$ on $\Fock_N(H)$ with semigroups inside $\Fock(H)$ so the first Trotter-Kato approximation can be used with $X=\Fock(H)$.
Since semigroups are contractive, $M=1, \omega=0$ can be fixed. With the identification under $\Pc_N, \Jc_N$, condition \emph{2)} of the Trotter-Kato approximation follows from the two claims below.
\begin{enumerate}[a)]
    \item $\Fock(H)^0$ is  dense inside $\Dom(\CarlemanOp{W_1,W_2})$ in the graph norm $\norm{\cdot}_{\CarlemanOp{W_1,W_2}}$, meaning $\Fock(H)^0$ is core for $\CarlemanOp{W_1,W_2}$.
    \item For each $x\in \Fock(H)^0$, there exists $x_n\in \Dom(A_n)$ such that $x_n\to x$ and $\CarlemanOp{W_1, W_2}_nx_n \to \CarlemanOp{W_1, W_2}x$.
\end{enumerate}

The second is straight-forward since if $x\in  \Fock(H)^0$ then the $k$-th tensor projection $\Tc_k(x)=0$ for all $k\geq K$, $K$ large enough. So one can choose the $x_i$ = $\Pc_ix$ if $i< K$, and $x_i=\Pc_K(x)=x$ for all $i\geq K$. Then ${(x_{i})}_{i\in \NN}\subset \Fock_K(H) \subset\cap_\NN \Dom(\CarlemanOp{W_1, W_2}_k)\cap \Dom(\CarlemanOp{W_1, W_2})$. As $\CarlemanOp{W_1, W_2}|_{\Fock_K(H)} = \CarlemanOp{W_1, W_2}_K$ which is a bounded operator, $$
\lim_{n}x_n\to x,\ \lim_{n}\CarlemanOp{W_1, W_2}_nx_n=\CarlemanOp{W_1, W_2}x.
$$

To see the first which is the graph-norm density of $\Fock(H)^0$ inside $\Fock(H)$, for $x\in \Dom(\CarlemanOp{W_1,W_2})$, choose $x_n=\Jc_n\Pc_n x\in \Fock(H)^0$. Then $\lim_n\CarlemanOp{W_1,W_2}_n x_n = \lim_n\CarlemanOp{W_1,W_2} x_n = \CarlemanOp{W_1,W_2}x$ since $\CarlemanOp{W_1,W_2}$ is closed
and
$\lim_n x_n = x$ with $x\in \Dom(\CarlemanOp{W_1,W_2})$.
Therefore, $(x_n)\subset \Fock(H)^0$ has limit $x$ in the graph norm: \aln{
\lim_n \norm{x-x_n}_{\CarlemanOp{W_1,W_2}} = \lim_n \norm{x-x_n} + \lim_n\norm{\CarlemanOp{W_1,W_2}x-\CarlemanOp{W_1,W_2}x_n} = 0 .
}
The two claims together are precisely condition $2)$ of the first Trotter-Kato approximation theorem; condition $2)$ is equivalent to condition $1)$ which is the uniform convergence of semigroups generated by $\Jc_NT_N\Pc_N$ on compact intervals as needed.
\end{proof}

\noindent From this we have an immediate result on convergence of truncated Carleman linearizations.

\begin{prop}\label{prop_carleman_convergence1} If $W_S\leq 0$ and $\Lambda_1\leq 0$ then $$\lim_{N\to \infty} \norm{\Tc_1 T_{t, N}\FExp_N(\phi_0) - u(t)}_H=0, \txt{ uniformly on } [0, T] \txt{ for any  }T>0,$$
where $u(t)$ is the solution to $\phi' = W_1\phi + W_2\phi^{\tsr 2}, \phi(0) = \phi_0, \norm{\phi_0} < 1$.
\end{prop}

\begin{proof} This follows since $\FExp(\phi_0)\in \Dom(\CarlemanOp{W_1,W_2})\in L^2(\Fock(H), Q)$, and by Theorem~\ref{theorem_dissipatinity_nonzero_tiling} 
along with Theorem~\ref{theorem_trotter_kato_applies}, $W_S, \Lambda_1\leq 0$ implies that the semigroup $T_{N}(t)$ generated $\CarlemanOp{W_1,W_2}_N$ converges to $T(t)$ generated by $\CarlemanOp{W_1,W_2}$ in $\norm{\cdot}_{\Fock(H)}$ and, therefore, in $\norm{\cdot}_H$, on compact time intervals. By Trotter-Kato theorem the error $e(N)$ converges uniformly to zero for $t\in [0, T]$, \aln{\lim_{N\to \infty} e(N) &=\norm{\Tc_1 T_{N}(t) \FExp_N(\phi_0) - u(t)}_H \\
&\leq \lim_{N\to \infty}\norm{T_{N}(t) \FExp_N(\phi_0) - T(t) \FExp(\phi_0)}_{\Fock(H)} = 0.
}
\end{proof}

\subsection{Comparison of dissipativity conditions}\label{section_convergence_comparison}

For $W_1$ self-adjoint, $W_1< 0$ bounded with largest eigenvalue $\lambda_1$,
then on the diagonal we have \aln{
\pmat{u\\ u^{\tsr 2}}^*&\pmat{W_1 & W_2^* \\ W_2^* &\smfrac{1}{2}\SymmSum_2(W_1)} \pmat{u\\ u^{\tsr 2}}\\
&= \inner{u, W_1u} + \smfrac{1}{2}\inner{u^{\tsr 2}, \SymmSum_2(W_1) u^{\tsr 2}} - 2\Re\inner{u, W_2 u^{\tsr 2}}\\
&\leq -|\lambda_1|(\norm{u}^2 + \norm{u}^4)   + 2\norm{W_2}\norm{u}^3 \\
&\leq -2|\lambda_1|\norm{u}\norm{u}^2 + 2\norm{W_2}\norm{u}^3 \leq 2(\norm{W_2}-|\lambda_1|  )\norm{u}^3 ,
}
where we used the inequality $-2ab\geq -a^2-b^2$ with $a,b\in \RR_\geq 0$.
Extending this, if we are interested in general, $\phi_0, \norm{\phi_0}=M$, then by  Remark~\ref{rem_adapted_system} \eqref{eq_adapted_system}, we need for any $\eps>0, M_\eps = M+\eps$,  $M_\eps\norm{W_2}-|\lambda_1| \leq 0$, equivalently $|\lambda_1| - M\norm{W_2}>0$. This is implied by the convergence condition $R<1$ noted in observation~\ref{obv_r_bound}, $\norm{W_2}\norm{\phi_0}/|\lambda_1|<1$.

Therefore, the dissipativity analysis is in agreement with the iterative argument on the diagonal. The obvious advantage of the dissipativity condition is the independence from operator norms, using remark~\ref{rem_dissipativity_on_hilbert_spaces}, we will push this to unbounded $W_i$'s. Since dissipativity is essential for contractive semigroups, this bound is likely tight if one wants both a $C_0$-semigroup and contractivity.

\section{A nonlinear Trotter-Kato approximation}

For 
the nonlinear system with bounded $W_i$'s, the convergence for $\CarlemanOp{W_1, W_2}_N$ to $\CarlemanOp{W_1, W_2}$ one only needs dissipativity. For $W_i$'s unbounded, additional hypotheses about the $W_i$'s cores are needed. First we consider types of nonlinearities for which the $2$-level dissipativity condition can hold. For $W_1:H\to H, W_2: H^{\tsr 2} \to H$. We will assume $W_1$ is self-adjoint with discrete spectrum. For every $a\oplus b\in H\oplus H^{\tsr 2}$, we need the bound \alnn{
\Re \inner{a\oplus b, A_W a\oplus b} :&= \smfrac{1}{2}\pmat{a\\ b}^*\pmat{W_1 & W_2 \\ W_2^* &  \smfrac{1}{2}\SymmSum_2(W_1) }\pmat{ a\\  b  }  \leq 0.
}
This comes down to a condition similar to relative-boundedness of the perturbations in linear semigroup theory
which is expected since $A_W$ is a perturbation of $W_1\oplus \SymmSum_2(W_1)/2$ by $W_2$, the former being a $C_0$-semigroup generator as $W_1$ is dissipative. Since $W_1\leq 0$ implies $\SymmSum_2(W_1)/2\leq 0$, there is a $Q$ so that $\SymmSum_2(W_1)/2 = -Q^*Q$.

\subsection{Nonlinear relative bounds}

\begin{prop}\label{prop_carleman_nonlinearity_for_convergence} Suppose that $\txt{Kernel}(\SymmSum_2(W_1)) \subset \txt{Kernel}(W_2), W_1\leq 0$. Then
$\Re \inner{w, A_W w}\leq 0$
holds for all $w=a\oplus b$, $a\in \Dom(W_1)$ and $b\in \Dom(W_2)\cap \Dom(\SymmSum_2(W_1))$ if for all $a\in \Dom(W_1)$, $|\inner{b, \SymmSum_2(W_1)b}|>0$ implies  \aln{
\inner{a, W_1 a} -\dfrac{(\Re\inner{a, W_2 b})^2}{\inner{b, \smfrac{1}{2}\SymmSum_2(W_1) b}} \leq 0.
}
Note the first and second term are both quadratic in $a$ while the numerator and denominator of the second term are quadratic in $b$, therefore, both $a, b$ can be normalized by any positive constant.
\end{prop}
\begin{proof} We have $
\Re \inner{a\oplus b, A_W a\oplus b}
= \inner{a, W_1 a } -2\Re\inner{a, W_2 b} + \inner{b, \smfrac{1}{2}\SymmSum_2(W_1)b}\label{eqn_dissipativity_inequality_needed}
$.
With $a, b$ fixed, $t\in \RR$, considering $\Re \inner{a\oplus tb, A_W a\oplus tb}$, define $f_{ab}:\RR\to \RR$, \aln{
f_{ab}(t) :&= \inner{a, W_1 a} -2t\Re\inner{a, W_2 b} + t^2\inner{b, \smfrac{1}{2}\SymmSum_2(W_1) b}.
}
We are interested in $f_{ab}(1)$. Since $\SymmSum_2(W_1)\leq 0$, $\inner{b, \SymmSum_2(W_1)b} = 0$ implies $\SymmSum_2(W_1)b=0$ and by hypothesis, $W_2b=0$. In which case $f_{ab}(t) = \inner{a, W_1a}\leq 0$. So we only need to consider $|\inner{b, \SymmSum_2(W_1)b}|>0$. Additionally, $f_{ab}(0), \lim_{t\to \pm \infty }f_{ab}(t) \leq 0$. Being quadratic in $t, f_{ab}$ has a max at $t^* = \Re\inner{a, W_2 b}/\inner{b, \smfrac{1}{2}\SymmSum_2(W_1) b}$ given by\alnn{
f_{ab}(t^*) &= \inner{a, W_1 a} -2\dfrac{\Re\inner{a, W_2 b}}{\inner{b, \smfrac{1}{2}\SymmSum_2(W_1) b}}\Re\inner{a, W_2 b} + \left(\dfrac{\Re\inner{a, W_2 b}}{{\inner{b, \smfrac{1}{2}\SymmSum_2(W_1) b}}}\right)^2\inner{b, \smfrac{1}{2}\SymmSum_2(W_1) b}\nonumber \\
&= \inner{a, W_1 a} -\dfrac{(\Re\inner{a, W_2 b})^2}{\inner{b, \smfrac{1}{2}\SymmSum_2(W_1) b}}.\label{eq_quadratic_bound1}
} Enforcing negative semi-definiteness for RHS to equation~\eqref{eq_quadratic_bound1} gives the claim.

\end{proof}

\noindent When $W_1<0$, the above condition reduces to a nonlinear relative bound.

\begin{corollary}
$W_1\leq 0$ with $\lambda_1=\inf\{|\lambda|:\lambda\in\txt{Spec}(W_1)\}, \lambda_1 >0$, and $W_2$ is relatively bounded in the sense that $
\norm{W_2b}^2  \leq |{{\inner{b, \smfrac{1}{2}\SymmSum_2(W_1) b}}}\lambda_1|.
$
implying $\norm{W_2b}\leq \norm{Qb}\sqrt{\lambda_1}$.
\end{corollary}

\begin{proof} Applying the hypothesis to $f_{ab}$, \eqref{eq_quadratic_bound1} and setting $f_{ab}\leq 0$ yields the claim, \alnn{
f_{ab}(t^*)&\leq -|\lambda_1|\norm{a}^2 + \left|\dfrac{\norm{a}^2\norm{W_2b}^2 }{\inner{b, \smfrac{1}{2}\SymmSum_2(W_1) b}}\right|\label{eq_quadratic_bound2} \leq 0.
}
\end{proof}

\noindent The type of perturbation we have is, therefore, an example of a nonlinear relatively bounded perturbation. In particular, it's not globally Lipschitz, which is often needed for handling nonlinear perturbations, for example, nonlinear Miyadera-Voight perturbations (see \cite{fkirine2022nonlinear}).

Now we consider an example where an infinite family of discretizations for a nonlinearity satisfying conditions from proposition~\ref{prop_carleman_nonlinearity_for_convergence} can be constructed.

\begin{example}[Dissipativity for the hyperviscous Burger's equation]\label{example_hyperviscuous}
Consider the hyperviscous Burger's equation $u_t = (-u^2)_x/2 + \nu \laplace^M u$ on $J=[0, 2\pi]$, $M$ odd, periodic boundary conditions. Let $(e_n)_{n\in \ZZ}, e_n = \exp(\I nx)/\sqrt{2\pi}$ be the standard Fourier basis for $H:=L^2(J)$ with $e_{mn}:=e_m\tsr e_n$ forming an orthonormal basis for $H^{\tsr 2}$. The equation can be put into the form $u_t = W_1 u + W_2u^{\tsr 2}$ where the action of $W_1, W_2$ on $ \sum_k a_k e_k\in H,  \sum_{mn} b_{mn} e_{mn}\in H^{\tsr 2}$ 
is defined by
\aln{
W_1 a &= - \nu \sum_k k^{2M} a_ke_k, \\
W_2 b &= - \smfrac{\I}{2\sqrt{2\pi}} \sum_{m,n} (m+n) b_{mn}e_{m+n}  = - \smfrac{\I}{2\sqrt{2\pi}} \sum_p \sum_{m+n=p} p b_{mn}e_{p}.
}
Note that $\txt{kernel}(W_1) = \{ e_0\}, \txt{kernel}(\SymmSum_2(W_1)) = \{ e_{00}\}\subset \txt{kernel}(W_2) $ and non-zero eigenvalues for $W_1$ are $k^2, k\in \NN$. Importantly, $\inner{e_0, W_2b}=0$ for all $b$. Assume $a_k, b_{mn}$ are finitely supported, $b\not\in \txt{Kernel}(\SymmSum_2(W_1))$, and set $a_\perp:=\sum_{k\neq 0} a_ke_k$. Then we have \aln{
\left|\dfrac{(\Re\inner{a, W_2b})^2}{\inner{b, \smfrac{1}{2}\SymmSum_2(W_1)b}}\right| &=  \left|\dfrac{(\Re\inner{\sum_k a_k e_k, - \smfrac{\I}{2\sqrt{2\pi}} \sum_p \sum_{m+n=p} p b_{mn}e_{p}})^2
}{
-\smfrac{\nu}{2}\sum_{m,n} b_{mn}^2(m^{2M}+n^{2M})
}\right|\\
&\leq \dfrac{\norm{a_\perp }^2}{\nu}\dfrac{\smfrac{1}{4 \cdot 2\pi} \sum_p\norm{ \sum_{m+n=p} p b_{mn}}^2}{\smfrac{1}{2}\sum_p \sum_{m+n=p} b_{mn}^2(m^{2M}+n^{2M})}.
}
Setting $K_{mn} =(m^{2M}+n^{2M})^{1/2}$, and applying Cauchy-Schwarz inequality to the sum over $m+n=p$, 
\aln{
&\dfrac{\smfrac{1}{2\cdot 2\pi} \sum_p\norm{ \sum_{m+n=p} (p/K_{mn})\cdot K_{mn}b_{mn}}^2}{\sum_l \sum_{m+n=l} b_{mn}^2(m^{2M}+n^{2M})}\\
&\hspace{10mm} \leq \smfrac{1}{4\pi}\sum_p \left(\dfrac{\sum_{m+n=p}(m^{2M}+n^{2M})^2b_{mn}^2}{\sum_l \sum_{m+n=l} b_{mn}^2(m^{2M}+n^{2M})} \cdot\sum_{m+n=p} \dfrac{p^2}{K_{mn}^2} \right)\\
&\hspace{10mm} \leq \smfrac{1}{4\pi}\sum_{p\in \ZZ} \sum_{m+n=p, m,n\in \ZZ}\dfrac{p^2}{m^{2M}+n^{2M}} := K_M.
}
Substituting $n=p-m$, one can verify that the inner sum behaves like $O(p^{3-2M})$ and $K_M$ is finite for $M>2$ (see appendix, proposition~\ref{prop_hyperviscous_cross_bound}). This means
for $\nu\geq \sqrt{K_M}$, we have the bound
\aln{
\inner{a, W_1a} - \dfrac{(\Re\inner{a, W_2b}^2)}{\inner{b, \smfrac{1}{2}\SymmSum_2(W_1)b}} \leq -\nu \norm{a_\perp}^2  + \dfrac{\norm{a_\perp}^2}{\nu} K_M \leq 0.
}
So the dissipativity condition is satisfied on the finite dimensional subspaces $H_{m} = \txt{LinSpan}\{e_k:|k|\leq m\}, a\in H_m, b\in H_m^{\tsr 2}$. Now with orthogonal projections, $P_m:H\to H_{m}$, for $n\in \NN$
define finite dimensional discretizations ${W_1}^{2n}:H_{2n}\to H_{2n}, {W_2}^{2n}:H^{\tsr 2}_{2n}\to H_{2n}$ as $W_1^{2n} = P_{2n} W_1 P_n$, $W_2^{2n} = P_{2n} W_2 P_n^{\tsr 2}$. Notice that they are just restrictions of $W_1, W_2$ to $H_n$, so for each $n$ they satisfy the dissipativity condition for truncated Carleman semigroup to be a contractive $C_0$-semigroup.
\end{example}

\subsection{Convergence for unbounded coefficients}
Consider the system, defined by unbounded operators $W_1:H\to H, W_2:H^{\tsr 2}\to H$, and their discretizations supported on $k$-dimensional subspaces $V_k\subset H,  F_1^k:V_k\to V_k, F_2^k:V_k^{\tsr 2}\to V_k$ for $k$ an increasing sequence going to infinity,  \aln{
\phi' = W_1\phi + W_2\phi^{\tsr 2}, \phi(0)=f,
\phi_k' = F^k_1\phi + F^k_2\phi^{\tsr 2}, \phi_k(0)=f_k.
}
We have the corresponding Carleman systems, $\FExp(\phi)' = \CarlemanOp{W_1, W_2}\FExp(\phi), \FExp(\phi(0))=\FExp(f)$, and $
\FExp(\phi_k)' = \CarlemanOp{F^k_1, F^k_2)}\FExp(\phi_k), \FExp(\phi_k(0))=\FExp(f_k)$. Suppose $\CarlemanOp{W_1, W_2}$ is closed with a core $D$, then the Trotter-Kato approximation theorem says that if for all $x\in D$, $\lim_{k\to \infty}\CarlemanOp{F^k_1, F^k_2)}x = \CarlemanOp{W_1, W_2}x$ and  $\CarlemanOp{F^k_1, F^k_2)}, \CarlemanOp{W_1, W_2}$  are $C_0$-semigroup generators on $\Fock(H)$, then the generated semigroups, $T_{t, k}, T_t$, converge. Now for the $2$-level systems, \aln{
\phi' = \pmat{W_1 & W_2 \\ 0 &\smfrac{1}{2}\SymmSum_2(W_1)}\phi := \Lambda_1\phi,\ \ \pmat{F^k_1 & F^k_2 \\ 0 &\smfrac{1}{2}\SymmSum_2(F^k_1)}\phi := \Lambda_1^k \phi.
}
Since ${F^k_i}'s, i\in [2]$ are bounded for each $k$, the hypothesis $\Lambda_i^k, F^k_1\leq 0$ imply that $\CarlFk{k}$ generates contractive $C_0$-semigroups, $T_{t, k}$. Now we need to consider the large $k$-limit, $\lim_{k\to \infty} \CarlFk{k}v$ for $v$ in an appropriately dense space $D$. For consistency $\lim_{k\to \infty} \CarlFk{k}v = \CarlemanOp{W_1, W_2}v$ is needed on a core $D$. But this is not enough to ensure that $\CarlemanOp{W_1, W_2}$ is a $C_0$-semigroup generator; unlike theorem~\ref{theorem_trotter_kato_applies} this still need to be established.
For a class of $F_j^k$'s it will turn out that $(\lim_{k\to \infty} \CarlFk{k})^*$ is dissipative, meaning checking the range condition in Lumer-Phillips theorem can be avoided. So we will additionally require  $\CarlFk{k})^*$ is dissipative.

\begin{remark} By the second Trotter-Kato theorem (see, ~\cite[Section~III.4.9]{engel_nagel_semigroups}), one can replace the dissipativity for $(\lim_{k\to \infty} \CarlFk{k})^*$ by convergence of the semigroups $T_{t, k}$ to a contractive $C_0$-semigroup or by the strong convergence of the resolvents $R(\lambda_0, \CarlFk{k})$ to an operator with dense range, but we will not consider this since it will usually require knowing $W_1, W_2$.
\end{remark}

\begin{theorem}[The dissipative adjoint Carleman-Trotter-Kato theorem]\label{thm_dissipative_carleman_trotter_kato} Suppose $\Lambda_1^k, F^k_1\leq 0$ for all $k$, the limit $\CarlemanOp{F_1, F_2} :=\closure{\lim_{k\to \infty}\CarlemanOp{F^k_1, F^k_2}}$ is densely defined, $\cap_k \Dom(\CarlemanOp{F^k_1, F^k_2})\cap \Dom(\CarlemanOp{F_1, F_2})$ dense, $\CarlW$ is densely defined and closed satisfying
\alnn{\label{eq_limit_trotter_kato_core_convergence}
\CarlemanOp{F_1, F_2} =\closure{\lim_{k\to \infty}\CarlemanOp{F^k_1, F^k_2}} = \CarlemanOp{W_1, W_2}.
}  Additionally, suppose $(\lim_{k\to \infty} \CarlFk{k})^*$ is dissipative. Then the Carleman operators $\CarlW$, $\CarlFk{k}$ generate contractive $C_0$-semigroups $T_t, T_{t, k}$, with $T_{t, k}$ converging to $T_t$ as in the Trotter-Kato approximation theorem.
\end{theorem}

\begin{proof} It's enough to show that $\lim_{k\to \infty}\CarlemanOp{F^k_1, F^k_2)}$ is dissipative, since being densely defined, with a dissipative adjoint,
it will generate a contractive $C_0$-semigroup by Lumer-Phillips theorem and Remark~\ref{rem_dissipativity_on_hilbert_spaces}. In particular it's closable, so the limit can be identified with its closure. Now $\Lambda_1^k, F_1^k\leq 0$ implies each $\CarlFk{k}$, and also $\CarlemanOp{F_1, F_2}$, are dissipative. Since $\closure{\lim_{k\to \infty}\CarlemanOp{F^k_1, F^k_2)}} = \CarlemanOp{W_1, W_2}$,
$\CarlW$ is a $C_0$-semigroup generator.
The convergence of the semigroups is exactly the first Trotter-Kato theorem.
\end{proof}

The condition \eqref{eq_limit_trotter_kato_core_convergence} is a amalgamation of convergence conditions for first and second Trotter-Kato theorems, and a form of consistency of generators. Now by Theorem~\ref{theorem_dissipatinity_nonzero_tiling}, $\CarlFk{k}^*$ are dissipative, so the limit $\CarlFk{k}^*$ will be dissipative when $\CarlFk{k}$'s are compatible in the right way. The following lemma to get at dissipativity is natural in the setting as the example~\ref{example_hyperviscuous_burgers_limit} will illustrate.

\begin{lemma}\label{lemma_CWadjoint_dissipativity} Let $A_n: \Dom(A_n)\subset V_n \to V_n\subset V$ be a sequence of operators on Hilbert subspaces $\dots V_n\subset V_{n+1}\dots\subset V$ of Hilbert space $V$. Suppose on $V_n, A_n, A_n^*$ are both dissipative. Extend $A_n$ to $V$ by setting $A_n(V_n^\perp)=0$ and let $D_n:= \Dom(A_n)$ after the extension. Additionally assume that $A_{n+1}|_{V_n} = A_n$. Define the limit operator $A$ with $D:=\Dom(A) = \bigcup_{n=1}^\infty D_n$ by $Ax = A_n x$ for $x \in D_n$. If $V=\closure{\cup_n V_n}$ and $A$ is densely defined, then $A, A^*$ are dissipative.
\end{lemma}

\begin{proof} The dissipativity of $A$ on $D$ is clear since if $x \in D$, $x \in D_n$ for some $n$ with $Ax=A_nx$ and $A_n$ is dissipative. Now $y \in \Dom(A^*)$ if there exists $g \in H$ satisfying $\inner{ Ax, y } = \inner{ x, g }$ for all $x\in D$ and $A^*y=g$. Since $D_n \subset D$, if this holds for all $x \in D$, it holds for all $x \in D_n$, implying
if $y \in  \Dom(A^*)$, then $y \in \Dom(A_n^*)$ and $\Dom(A^*)\subset \cap_n \Dom(A_n^*)$.

Let $x\in \Dom (A)$, with $n$ such that $x\in D_n$ (so for orthogonal projections $P_n:V\to V_n$, $P_nx = x$), we have $Ax\in V_n$, meaning $\inner{Ax,y} =  \inner{P_nx,A^*y} = \inner{x,P_nA^*y}$, and additionally
\aln{
\inner{Ax,y} 
&= \inner{AP_n x,y} = \inner{A_nP_n x,P_n y}=\inner{P_nx,A_n^* P_n y} = \inner{x,A_n^* P_n y}.
}
Since this holds for all $x\in D$ which is dense, $P_n A^* y= A_nP_n y $. Fix $y \in  \Dom(A^*)$, then $V=\closure{\cup_n V_n}$ implies that  $P_n$ converges strongly to $I$, giving $$
\Re \inner{A^*y,y} = \lim_{n\to \infty} \inner{P_nA^* y,P_n y}  = \lim_{n\to \infty} \inner{A_n^*P_n y,P_n y} \leq 0 .
$$
using $A_n^*$ is dissipative on $V_n$.

\end{proof}

\begin{example}\label{example_hyperviscuous_burgers_limit} Consider the discretization in example~\ref{example_hyperviscuous}, $W^{2n}_i:i\in[2]$. The  hypothesis of the Lemma~\ref{lemma_CWadjoint_dissipativity} are satisfied with $V_{n}=\Fock(H_{2n})$, $A_n=\CarlemanOp{W^{2n}_i:i\in[2]}$. So for the limit $A$, $A^*$ is dissipative. In general, it holds for spectral discretizations.
\end{example}

It remains to characterize the core since it's enough to check the equality 
\refeq{eq_limit_trotter_kato_core_convergence}, $
\closure{\lim_{k\to \infty}\CarlemanOp{F^k_1, F^k_2)}} = \CarlemanOp{W_1, W_2}$ on a core. For the quadratic system considered, $W_i$'s preserve or lower the tensor degree, they correspond to preservation and annihilation operators on Fock space.   Observe that for every $k$, the Carleman operator $\CarlFk{k}$ restricted to the space $\Fock(H)_n$, agrees with $\CarlFk{k}_n$, and so is generator of a $C_0$-semigroups $T_{t, n}$ on $\Fock(H)_n$. This means the semigroup $T_t$ generated by $\CarlemanOp{F_1, F_2}$ on $\Fock(H)$, satisfies $T_{t}(\Fock(H)_n)\subset \Fock(H)_n$ for all $n$. Now because $\cup_{n,k} \Dom(\CarlFk{k})$ is dense, and invariant under $T_t$, it's a core by \cite [Proposition~II.1.7]{engel_nagel_semigroups}.

\begin{corollary}\label{cor_core_is_finite_particles} For example~\ref{example_hyperviscuous_burgers_limit}, with $H=L^2(S^1), \Fock(H)^0$, contains a core for $\CarlemanOp{F_1, F_2}$,  $\Fock(H)^0\cap (\cup_{n,k} \Dom(\CarlFk{k}_n)).$
\end{corollary}

\begin{remark} The utility of the finite particle space containing a core is that instead of considering the convergence for the untruncated $\CarlFk{k}$, only the convergence for the truncations $\CarlFk{k}_{m_k}$ for any increasing sequence $(m_k)$ tending to infinity needs to be considered. Finally, being a contractive $C_0$-semigroup generator $\CarlemanOp{F_1, F_2}$ is maximal dissipative (see \cite{arendt2023extensions}), so the equality  $\CarlemanOp{F_1, F_2}=\CarlemanOp{W_1, W_2}$ follows if $\CarlemanOp{W_1, W_2}$ is also dissipative on its core which can be characterized similar to corollary~\ref{cor_core_is_finite_particles}.
\end{remark}

\section{Nonlinear perturbations and integrated semigroups}\label{nonlinear_perturbations}

For a polynomial system on Hilbert space $H$ \alnn{\label{eq_poly_nonlinear_system}
x'(t) = W_1 x + W_2 x^{\tsr 2} + \ldots +  W_p x^{\tsr p}, x(0):=x_0
}
the Carleman system is given by \begin{equation}
\left\{
\begin{aligned}
y'(t) &= \CarlemanOp{W_j,j\in [p]} y(t), \\
y(0) &= \FExp(x_0)
\end{aligned}
\right.
\label{eq:InfSysMatrix}
\end{equation}
where $\CarlemanOp{W_j,j\in [p]}$ is the infinite-dimensional block upper-triangular operator (see, for example, \cite{forets2017explicit})
\alnn{
&\CarlemanOp{W_i:i\in [p]} := A + B\ \label{eq_perturbation_rep}  \\
&\hspace{1mm}= \begin{pmatrix}
W_1^1 & 0 & 0 & \ldots & 0 & 0 & 0 & \ldots \\
0& W_2^2 & 0 & \ldots & 0 & 0 &0 & \ldots \\
0 & 0 & W_3^3 & \ldots  & 0 & 0  & 0 & \ldots \\
\vdots &  \vdots &  \vdots &    & \vdots &\vdots & \vdots &\\
\end{pmatrix}\nonumber \\
&\hspace{5mm}+ \begin{pmatrix}
0 & W_2^1 & W_3^1 & \ldots & W_{k}^1 & 0 & 0 & \ldots \\
0& 0 & W_3^2 & \ldots & W_{k}^2 & W_{k+1}^2&0 & \ldots \\
0 & 0 & 0 & \ldots  & W_{k}^3 & W_{k+1}^3  & W_{k+2}^3 & \ldots \\
\vdots &  \vdots &  \vdots &    & \vdots &\vdots & \vdots &\\
\end{pmatrix} ,\nonumber
}
with $W^i_{i+j-1} := \SymmSum_i(W_j)$.
We require that as for the quadratic nonlinearity in equation~\eqref{eq_orig_nonlinear_eqn}, $W_i$'s are symmetric. The Carleman operator $\CarlemanOp{W_i, i\in [p]}$ can be viewed as a perturbation of the block diagonal operator $A:=\txt{Diag}(A^i_i)_{i\in \NN \nonumber}$ by the linear operator $B$.  Now $W_1\leq 0$ is necessary: if $W_1$ is self-adjoint and diagonalizable with discrete spectrum, then $\SymmSum_n(W_1)$ has $n\lambda$ as an eigenvalue for any eigenvalue $\lambda$ for $W_1, W_1 v_i = \lambda v_i$, and $v_i^{\tsr n}$ is the corresponding eigenvector. Therefore, if $W_1\leq 0$ does not hold, then $\CarlemanOp{W_1}$ cannot generate a $C_0$-semigroup, not just a contractive $C_0$-semigroup, and perturbation theory doesn't apply. Note that we don't need $\dim H$ to be finite.

However, while $W_1$ must be dissipative, perturbations of $W_1$ by bounded operators, and nonlinear functions of bounded operators, can be
nondissipative on $\Fock(H)$.  This has been observed for reaction-diffusion equations by \cite{liu2023efficient}; next we consider polynomial perturbations that generalize the convergence results for reaction-diffusion equations.

\begin{remark}
It will be convenient to work in a basis in which $W_1$ is diagonal since this has no effect on existence of the resolvent and its bounds, $\norm{R(A, \lambda)} = \norm{\psi^{-1}(\lambda - A)^{-1}\psi^{}}
$
where $\psi$ is unitary that diagonalizes $W_1$. If $\psi$ is not unitary then the resolvent bounds are scaled by the condition number for $\psi$, but the existence is unaffected.
\end{remark}

\begin{remark} While $x'(t) = W_0(t) + W_1 x + W_2 x^{\tsr 2} + \ldots +  W_p x^{\tsr p}$ has a similar Carleman linearized system with $\SymmSum_n(W_0(t))$ appearing below the diagonal, we do not consider this inhomogeneity in this article. The reason being $W_0$ is a creation operator on $\Fock(H)$ while $W_i, i>1$ are annihilation operators. The creation operator means that $\CarlemanOp{W_0}$ maps $\Fock(H)_N$ to $\Fock(H)_{N+1}$, which has an effect on the core for $\CarlemanOp{W_i:i\in[0:p]}$. Additionally, while a polynomial nonlinearity of order $p$ can be reduced to a quadratic nonlinearity on a larger Hilbert space, we do need $W_i$'s to be symmetric. The analysis is convenient enough that we do not need to explicitly reduce to a quadratic system, and check if the system can be symmetrized.
\end{remark}

\subsection{Relatively bounded perturbations}

Suppose $W_1$ is densely defined and closed, diagonal with eigenvalues, $(\lambda_i)$ ordered as $0\geq \dots \lambda_i \geq \lambda_{i+1}$ associated to eigenbasis $(w_i)_{i\in \NN}, W_1w_i = \lambda_i w_i$, and is strictly dissipative, $\lambda_1 < 0$. And let $W_j, j>1$ be bounded operators with $\gamma=\sup_{j>1} \norm{W_j}$. Now we need to establish first that $\Dom(A)\subset \Dom(B)$, so $A+B$ is defined on a dense subspace. For this, $|\lambda_1| > 0$ is useful. It will turn out that $B$ is $A$-bounded.

\begin{definition} For linear operators $S: \Dom(S) \subset X\to X, S': \Dom(S') \subset X\to X$, $S'$  is called (relatively)
$S$-bounded if $\Dom(S) \subset \Dom(S')$ and there exist constants $a,b\in \RR^{\geq 0}, \norm{S'x} \leq a \norm{Sx} + b\norm{x}$
for all $x \in \Dom(S)$. The $S$-bound of $S'$ is the minimum such  $a$ over all choices of $b$.  
\end{definition}

\begin{remark}\label{rem_kato_closable}
    By \cite[Theorem~IV.1.1]{kato_perturbations}, if $S'$ is $S$-bounded with $S$-bound $a<1$ then $S+S'$ is closable iff $S$ is closable.
\end{remark}

\begin{remark}\label{rem_relative_bounded_core_closable} The Banach space $X$ can be normed using $\norm{x}_1 = \norm{x} + \norm{Sx}$ similar to Sobolev norms, $X_1:=(\Dom(S), \norm{\cdot}_1)$ being the Sobolev space relative to $S$. Now $S'\in \Bc(X_1, X)$ means $\norm{S'}_{\Bc(X_1, X)}\leq \alpha <\infty$, $$
\norm{S'x} \leq  \alpha \norm{x}_1 =  \alpha(\norm{Sx} + \norm{x}) < \infty,
$$
therefore, if $S'$ is $S$-bounded, then $S'\in \Bc(X_1, X)$. So it's enough that $S$-boundedness holds on a core for $S$, since then $S'$ is a densely defined, bounded operator $S':\Dom(S)\to  X$ and extends uniquely to all of $\Dom(S)$.
\end{remark}

\begin{prop}\label{prop_a_boundedness} For $A, B$ in equation~\eqref{eq_perturbation_rep}, $B$ is $A$-bounded with $A$-bound $(k-1)\gamma/|\lambda_1|$, so by remark~\ref{rem_kato_closable}, when $(k-1)\gamma/|\lambda_1|<1$, $A+B$ is closable.
\end{prop}

\begin{proof} Suppose $\V u:=(u_m)_{m\in \NN}\in \Dom(A)$ where $u_m\in H^{\tsr m}$, then \alnn{\label{eq:lower_bound_w1_fock}
\eta =\sum_{m\in \NN} m^2|\lambda_1|^2\norm{u_m}^2 \leq  \sum_{m\in \NN} \norm{\SymmSum_m(W_1)u_m}^2 = \norm{A\V u}^2 < \infty .
}
Using that $B$ can be written as sum over operator mapping in to $m$-particle spaces, the $m^{th}$-component of $B\V u, (B\V u)_m$ is bounded by $$
\norm{(B\V u)_m} = \|\sum_{j=1}^{k-1} \SymmSum_m(W_{j+1})u_{m+j}\| \leq \sum_{j=1}^{k-1} m \gamma \norm{u_{m+j}} \leq m \gamma \sqrt{(k-1) \sum_{j=1}^{k-1}  \|u_{m+j}\|^2 } ,
$$
where for the last bound Cauchy-Schwarz was used. Therefore, \alnn{
\|Bu\|^2 &\leq (k-1) \gamma^2 \sum_m \sum_{j=1}^{k-1} m^2\|u_{m+j}\|^2 \leq (k-1)^2\gamma^2 \sum_m m^2\|u_{m}\|^2\\
&\leq \dfrac{(k-1)^2\gamma ^2}{|\lambda_1|^2}\norm{Au}^2,
 \label{eq:B_is_A_bounded}
}
from which the conclusion follows.
\end{proof}

Now we are interested in conditions under which $A+B$ generates a $C_0$-semigroup. By \cite[theorem~III.2.7]{engel_nagel_semigroups}, if $S$ is a generator of a contractive $C_0$ semigroup, $S'$ is $S$-bounded and dissipative, with $S$-bound $a_0<1$. Then  $(S + S', \Dom(S))$ is a generator of a contractive $C_0$-semigroup. Equivalently, this also follows by the Lumer-Phillips theorem by considering the dissipativity for $A+B$ since $A$ is already dissipative. It is possible to relax the condition $\lambda_1<0$ to $\lambda_1\leq 0$ if $\txt{kernel}(A)\subset\txt{kernel}(B)$ which can be checked depending on $W_i$'s. Since if $\txt{kernel}(A) \subset \txt{kernel}(B)$, then with $P_{0}, P_1$ being projections onto $\txt{kernel}(B), \txt{kernel}(B)^\perp$,
we have $\norm{Bu} = \norm{BP_1u}$. Now $A_{|\txt{kernel}(B)^\perp}$ is strictly negative, so the role of $\lambda_1$ is now played by $\lambda_2$. However, a polynomial perturbation without a degree zero term will not be dissipative. For some $W_i$'s it is possible to examine the spectrum to show that $A+B$ is a $C_0$-semigroup generator, but for the general case we need to pass to $1$-integrated semigroups.

The requirement that the relative $A$-bound $a_0<1$, which is the same as saying that $|\lambda_1|$ is relatively large (or under kernel inclusion, $|\lambda_2|$ is large), appears frequently in the literature, for example, the perturbation $B$ is similar to reaction-diffusion equations considered in \cite{liu2023efficient}, where they analyze $u' = (F_1 + a \one )u + bu^{\tsr M}$ and require the constraint on $b/|\max \txt{spec}(F_1)|$ for convergence. The key significance of this constraint is closability in remark~\ref{rem_kato_closable}.

\subsection{Nonlinear perturbations and integrated semigroups}

For the system, $\phi' = W_1\phi + W_p\phi^{\tsr p}$ with $W_1$ self-adjoint and negative semi-definite with largest eigenvalue $-|\alpha|$.
Now from equation~\eqref{eq:B_is_A_bounded}, if $W_p$ is bounded then $\CarlemanOp{W_p}$ also has $\CarlemanOp{W_1}$-bound $\norm{W_p}/|\alpha|$ where $b=0$. Under the hypothesis on $W_1$, $A=\CarlemanOp{W_1}$ generates a contractive $C_0$-semigroup, so by the Feller-Miyadera-Phillips Theorem,  the resolvent $R(\lambda, A) = (\lambda I - A)^{-1}$ exists for all $\lambda > 0$ and satisfies $
\norm{R(\lambda, A)} \leq 1/\lambda
$. Notice that we do not need $W_1$ to be bounded. We will use the following estimate for resolvent $R(\lambda, A+B)$ for any $B$ that is $A$-bounded.
\begin{prop}\label{prop_resolvent_bound_a_bounded_perturbation} For $\lambda > b/(1-2a):=\omega, M=1/(1-2a)$, $a<1/2$, where $a, b$ are constants that realize the $A$-bound for $B$, then we have
$
\norm{R(\lambda, A+B)} \leq {M}/{\lambda - \omega}.
$ 
\end{prop}

\begin{proof}
We have  $
    R(\lambda, A+B) =  (\lambda - A - B)^{-1} = ((1- BR(\lambda, A)) (\lambda - A))^{-1} =R(\lambda, A)(1-BR(\lambda, A))^{-1}
$.
Since $B$ is $A$-bounded, with constants $a, b \ge 0$, for all $x \in D(A)$, $\norm{Bx} \leq a\norm{Ax} + b\norm{x}$, \aln{
\norm{B R(\lambda, A)x} &\leq a \norm{A R(\lambda, A)x} + b\norm{R(\lambda, A)x}.
}
Suppose $a<1/2$. Then for $\lambda > \lambda_0 = b/(1-2a)$, using $A R(\lambda, A) = \lambda R(\lambda, A) - I$ we obtain
\aln{
\norm{B R(\lambda, A)} &\leq a \norm{\lambda R(\lambda, A) - I} + b\norm{R(\lambda, A)} \leq 2a + {b}{\lambda}^{-1} < 1.
}
Therefore, the Neumann series gives the inverse, and the needed growth bound,\aln{
\norm{(1-BR(\lambda, A))^{-1}} &= \norm{\sum_{k=0}^\infty BR(\lambda, A)^k}
\leq \sum_k\left(2a+{b}{\lambda}^{-1}\right)^k = \dfrac{1}{1-\left(2a +b{\lambda^{-1}}\right)},\\
\norm{R(\lambda, A+B)} &\leq \dfrac{1}{\lambda} \dfrac{\lambda}{(1-2a)\lambda - b}
=  \dfrac{M}{\lambda - \omega}. 
}

\end{proof}

Note the bound used above $\norm{\lambda R(\lambda, A) - I}\leq 2$ can be tightened for well behaved semigroups -- analytic or with bounded generators, by taking $\omega$ larger. Since $M>1$ Lumer-Phillips does not apply, and without checking powers of the resolvent Feller-Miyadera-Phillips may not either, but the bound \eqref{eq_resolvent_bound_once_integrated_semigroup} does apply. So with this estimate we can show as an example that the Carleman operator for reaction-diffusion equations is a generator of at least a $1$-integrated semigroup. Solving the linearized equation is reduced to approximating the $1$-integrated semigroup which can be done via Trotter-Kato type approximations, where convergence of generators to $1$-inrtegrated semigroups on a core implies convergence of the semigroups, for example, see \cite[corollary~3.5]{xiao2000_integrated_approximations}.

\begin{example} Consider the reaction-diffusion equation of type $
\phi' = (W_1-\alpha)\phi + W_p \phi^{\tsr p} $
on $J=[0, T]$ where $W_1=\laplace$ is the Laplacian, $\norm{W_p}<\infty$. Then by Lumer-Phillips $\CarlemanOp{W_1-\alpha}$ generates a contractive $C_0$-semigroup on $\Fock(H), H=L^2([0, T])$. Now if the $\CarlemanOp{W_1}$-bound $a$ on $\CarlemanOp{W_p}$ is less than $1/2$, which is true when $\norm{W_p} < \alpha/2$,  then
using Proposition~\ref{prop_resolvent_bound_a_bounded_perturbation},
\aln{
R(\lambda, A+B)\leq \dfrac{1}{\lambda}\dfrac{1}{1-2a},
}
implying the bound~\eqref{eq_resolvent_bound_once_integrated_semigroup} holds and $\CarlemanOp{W_1-\alpha, W_p}$ is a generator for $1$-integrated semigroup. Therefore, the reaction diffusion equation for such $\alpha$ is well-posed in the mild-solution sense, and the solution can be approximated from the semigroup.
\end{example}

Note that in the above, the boundedness of $W_2$ is not necessary if one can directly show $\CarlemanOp{W_2}$ is $\CarlemanOp{W_1}$-bounded.

\section{Summary and discussion}

In this article, Carleman linearization has been approached using semigroup theory. We treat quadratic nonlinear systems, \eqref{eq_orig_nonlinear_eqn}, by appealing to dissipativity, deriving a condition on a finite level truncated Carleman system that guarantees the dissipativity of the untruncated system. This condition for bounded nonlinearities is closely related to a condition well studied in quantum computing literature. However, we are also able to obtain convergence when both $W_1$ and $W_2$ are unbounded, and additionally, the linear part, $W_1$,  has a zero eigenvalue. This is seen to be a type of nonlinear relative bound. The convergence is then obtained through the first Trotter-Kato approximation theorem. 

For general polynomial systems, \eqref{eq_poly_nonlinear_system}, we give conditions under which the semigroup is generically a $1$-integrated semigroup. The question of convergence and the convergence similarly being addressed by variants of the Trotter-Kato approximation theorem. By establishing conditions for convergence we fill a gap in the functional analytic foundation of Carleman linearization.

\appendix
\section{Appendix}
\begin{prop}For $M>2$, \label{prop_hyperviscous_cross_bound} \aln{ K_M &= \dfrac{1}{4\pi}\sum_{p\in \ZZ }\sum_{m+n=p}\dfrac{p^2}{m^{2M}+n^{2M}}  \\
&\leq \dfrac{1}{4\pi}\sum_{p\in \ZZ, p\neq 0}p^2\sum_{m\in \ZZ} \dfrac{1}{m^{2M} + (p-m)^{2M}} := \dfrac{1}{4\pi}\sum_{p\neq 0} p^2 K_{M, p}
< \infty.
}
\end{prop}

\begin{proof} 
For $|m|\geq |p|$, the terms in $K_{M, p}$ are decreasing with increasing $|m|$, so these can be bound by integrals, \aln{
K_{M, p}
&\leq \sum_{|m|\leq|p|} \dfrac{1}{m^{2M}+ (p-m)^{2M}} + \int_{|p|}^\infty \dfrac{1}{m^{2M} + (p-m)^{2M}} dm\\ &\hspace{2mm}+ \int^{-|p|}_{-\infty} \dfrac{1}{m^{2M} + (p-m)^{2M}} dm.
}
Since  $|p| \geq 1$, with substitutions $pw=m, |p|y=m, |p|z = m$ 
\aln{
K_{M, p}&\leq \dfrac{1}{p^{2M}}\sum_{|m|\leq |p|, w=m/p} \dfrac{1}{w^{2M} + (1-w)^{2M}} + \dfrac{1}{p^{2M}}\int_{1}^\infty \dfrac{1}{y^{2M} + (1-y)^{2M}} |p|dy\\  &\hspace{0.5in} + \dfrac{1}{p^{2M}}\int^{-1}_{-\infty} \dfrac{1}{z^{2M} + (1-z)^{2M}} |p| dz \\
&\leq \dfrac{1}{p^{2M}}\left( 4^M(2|p|+1) + |p|\int_{1}^\infty \dfrac{1}{y^{2M} + (1-y)^{2M}} dy \right.\\
&\hspace{5mm}\left.+ |p|\int^{-1}_{-\infty} \dfrac{1}{z^{2M} + (1-z)^{2M}}  dz  \right) 
}
using that one of $w, 1-w\geq 1/2$ for each term in the sum over $w$. 

The integrals decay like $y^{-2M},z^{-2M}$ for $y,z$ large, with integrands that are positive, continuous, bounded away from zero, and independent of $p$, therefore, they are bounded uniformly in $p$, giving that there exists $K_1$ satisfying
\aln{
K_{M, p} &\leq \dfrac{|p|}{p^{2M}}\left(4^{M+1} + \int_{1}^\infty \dfrac{1}{y^{2M} + (1-y)^{2M}} dy + \int^{-1}_{-\infty} \dfrac{1}{z^{2M} + (1-z)^{2M}}  dz\right)\\
&\leq \dfrac{|p|}{p^{2M}}K_1.
} This implies for $M>2$, $
K_M \leq \smfrac{1}{4\pi} K_1 \sum_{p\in \ZZ, p\neq 0}{|p|p^2}/{p^{2M}} <\infty.
$ as needed. 

\end{proof}

\emph{Acknowledgment: The first author would like to thank Jun Liu for thoughtful comments.}
\printbibliography

\end{document}